\documentclass[12pt,journal,draftclsnofoot,onecolumn,english]{IEEEtran}
\usepackage{graphicx}
\usepackage{epstopdf}
\usepackage{amsmath}
\usepackage{amsthm}
\usepackage{amsfonts}
\usepackage[justification=centering]{caption}
\usepackage{color}
\newtheorem{theorem}{Theorem}
\newtheorem{lemma}{Lemma}
\newtheorem{definition}{Definition}
\newtheorem{corollary}{Corollary}
\newtheorem{example}{Example}
\newtheorem{remark}{Remark}

\newtheorem{proposition}{Proposition}
\usepackage{bm}
\usepackage{booktabs}
\usepackage{multirow}
\usepackage{algorithm}
\usepackage{algpseudocode}
\usepackage{amssymb}
\usepackage{cite}
\usepackage{pstricks}
\usepackage{hyperref}
\usepackage{caption}
\usepackage{subcaption}
 \usepackage{lipsum}
 \usepackage{aecompl}
 \usepackage{arydshln}
\newcommand{\tabcaption}{\def\@captype{table}\caption}
\newcommand{\tabincell}[2]{\begin{tabular}{@{}#1@{}}#2\end{tabular}}
\begin{document}
\title{Coded Caching Schemes for Multiaccess Topologies via Combinatorial Design}
\author{Minquan~Cheng,  Kai~Wan,~\IEEEmembership{Member,~IEEE,} Petros Elia,~\IEEEmembership{Member,~IEEE}
and~Giuseppe~Caire,~\IEEEmembership{Fellow,~IEEE}

\thanks{
A short version of this paper   was presented at the 2023 IEEE International Symposium on Information Theory (ISIT).
}
\thanks{M. Cheng is with Guangxi Key Lab of Multi-source Information Mining $\&$ Security, Guangxi Normal University,
Guilin 541004, China  (e-mail:  chengqinshi@hotmail.com). The work of M.~Cheng was in part supported by 2022GXNSFDA035087, NSFC (No.62061004, U21A20474), Guangxi Collaborative Innovation Center of Multi-source Information Integration and Intelligent
Processing, the Guangxi Bagui Scholar Teams for Innovation and Research
Project, and the Guangxi Talent Highland Project of Big Data Intelligence
and Application.}
\thanks{K.~Wan is with the School of Electronic Information and Communications,
Huazhong University of Science and Technology, 430074  Wuhan, China,  (e-mail: kai\_wan@hust.edu.cn). The work of K.~Wan  was partially funded by the   National Natural
Science Foundation of China (NSFC-12141107).}
\thanks{Petros Elia is with the Communication Systems Department, EURECOM, 06410 Sophia Antipolis, France (e-mail:
elia@eurecom.fr). The work of P.~Elia was supported in part by  the ERC Project DUALITY under Grant
725929, as well as by a Huawei France-funded Chair towards Future Wireless Networks.}
\thanks{G. Caire is with the Electrical Engineering and Computer Science Department, Technische Universit\"{a}t Berlin,
10587 Berlin, Germany (e-mail: caire@tu-berlin.de). The work of G.~Caire was partially funded by the European Research Council under the ERC Advanced Grant N. 789190, CARENET.}
}
\maketitle

\begin{abstract}
This paper studies a multiaccess coded caching (MACC) where the connectivity topology between the users and the caches can be described by a class of combinatorial designs. Our model includes as special cases several MACC topologies considered in previous works. The considered MACC  network  includes a server containing $N$ files, $\Gamma$ cache nodes and $K$ cacheless users, where each user can access  $L$  cache nodes. The server is connected to the users via an error-free shared link, while the users can retrieve the cache content of the connected cache-nodes while the users can directly access the content in their connected cache-nodes. Our goal is to minimise the worst-case transmission load on the shared link in the delivery phase. 
  The main limitation of the existing MACC works is that  only some specific access topologies are considered, and thus the number of users $K$ should be  either
    linear or exponential to $\Gamma$.  We overcome this limitation by formulating a new access topology derived from two classical combinatorial structures, referred to as the $t$-design and the $t$-group divisible design. In these topologies, $K$ scales linearly, polynomially, or even exponentially with $\Gamma$. By leveraging the properties of the considered combinatorial structures, we propose two classes of coded caching schemes for a flexible number of users, where the number of users can scale linearly, polynomially or exponentially with the number of cache nodes. In addition, our schemes can unify most schemes for the shared link network and unify many schemes for the multi-access network except for the cyclic wrap-around topology.
%
%
%
%
%
%
 \end{abstract}


%
%

\section{Introduction}


Caching can effectively shift traffic from peak to off-peak times \cite{BBD} by storing fractions of popular content in users' local caches during peak traffic  times, so that users can be partially served from their local caches, thereby reducing
network traffic.
In the setting of the shared-link caching model, a single server has access to a library of $N$ files of equal length and   serves  $K$   users through an error-free shared link, where each user has a    cache of    $M$ files.
   A   coded caching scheme contains two phases: the placement phase and the delivery phase. During the placement phase, each cache is filled with content from the library without any knowledge of future users' requests. During the delivery phase, each user requests a single file from the library. In terms of the users' requests and caches, the server then broadcasts a sequence of    messages such that each user's request can be satisfied.
The objective is to minimize the worst-case number of transmissions during the delivery phase (normalized by the file size)    among all possible requests, referred to as worst-case load $R$. Compared to the conventional uncoded caching scheme, the seminal coded caching  technique was originally proposed by Maddah-Ali and Niesen (MN) in~\cite{MN} to further reduce the number of transmissions during peak traffic times by introducing an additional coded caching gain.
The  coded caching scheme proposed by MN (referred to as MN scheme) uses a combinatorial cache placement and transmits coded multicast messages during the delivery where each multicast message is simultaneously useful to a set of users.
The  MN scheme in \cite{MN}  is generally order optimal within a factor of $4$~\cite{GR}, and is
 optimal under the constraint  of $N\geq K$ and uncoded cache placement (i.e., each user directly stores some bits of files in the library)~\cite{WTP,YMA}.
 For the case   $N< K$, the MN scheme was further improved by removing redundant multicast messages~\cite{YMA}; the resulting scheme is generally order optimal within a factor of $2$~\cite{yufactor2TIT2018} and is
 optimal under the constraint  of uncoded cache placement\cite{YMA}.
Combinatorial design was introduced into coded caching to construct coded caching schemes in \cite{YCTC}, where
a combinatorial structure, referred to as    \emph{Placement Delivery Array (PDA)}, was introduced into coded caching to construct coded caching schemes with the uncoded placement and one-shot delivery. It was shown in \cite{STD,SZG}  that the coded caching schemes in \cite{MN,YCTC,STD,SDLT,SZG,TR,YTCC,CKSB} can be represented by appropriate PDAs.
Under the PDA structure, various coded caching schemes   were then designed; just list a few,~\cite{YCTC,CJWY,CJTY,CJYT,MW,ZCJ,YTCC,SZG,CWLZC,SB,ER,ENR}.

Following the seminal results in \cite{MN}, numerous works have addressed  a variety of extended models beyond the original shared-link setting.
While the mature literature on coded caching has explored a rich variety of settings, some recent results - which we discuss later - have brought to the fore a powerful new way of exploiting a modest number of edge caches, with the motivation of edge caching which was widely used in the communication systems.   
 In this paper we consider this multiaccess coded caching (MACC) model, where cache content is stored at edge cache-nodes in the network and users do not have their own caches, as shown in Fig.~\ref{multiaccess-system}. The original MACC model
was proposed in  \cite{JHNS},
 containing a   server with $N$ files,   $\Gamma$   cache-nodes and $K$ \emph{cache-less} users.
Each   user could access a subset of cache-nodes with negligible cost in a cyclic wrap-around way; i.e,   each user can access $L$ neighbouring cache-nodes in a cyclic wrap-around fashion.
  An MACC scheme also contains two phases. During the placement phase, each cache-node places  $M$ files in its memory, with full knowledge of the network topology but without knowledge of subsequent users' requests. During the delivery phase, each user requests a single file; the server then broadcasts messages to the users such that each user can recover its requested file from the broadcasted messages and the cache content in the cache-nodes which it can access.
  As in the shared-link model, the objective is to reduce the worst-case delivery load $R$ for any given memory size $M$.
\begin{figure}
\centerline{\includegraphics[scale=0.5]{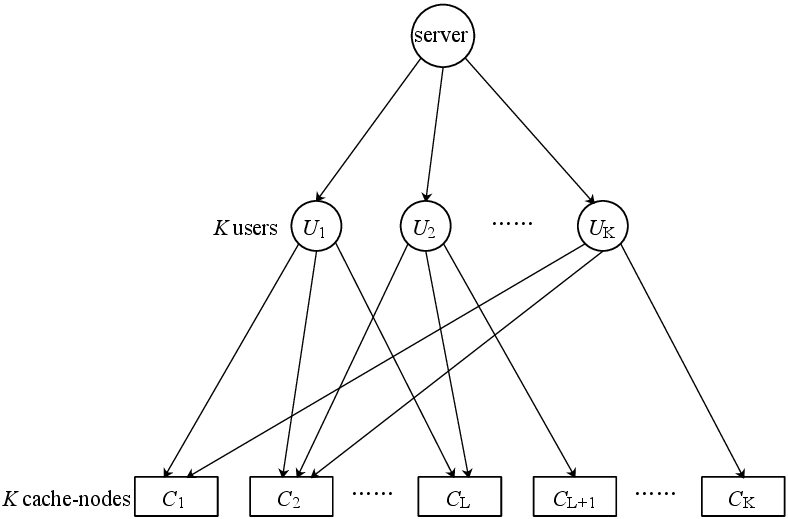}}
\caption{\small Multiaccess  coded caching system.}
\label{multiaccess-system}
\end{figure}

The authors in \cite{JHNS} proposed   a   coloring-based scheme that   allows for the maximum local caching gain  (i.e., the cache content stored by the cache-nodes connected to one user are totally different and thus each user can totally retrieve $LM$ files from the connected cache), but the coded caching gain is even worse than
the MN scheme for the shared-link model where each user caches $M$ files.
Recent works \cite{CWLZC,SB,MARB,SPE,RK,SR,RKstructure,KMR,WCLC,ZWCC,WCKC}   proposed   improved MACC schemes and converse bounds for this
cyclic topology. However, the general order optimality (with or without the constraint of uncoded cache placement) for this model still remains open.


Recently, the MACC problem with   combinatorial access topology  was considered in \cite{MKRmn}. This topology,
 involves $\Gamma$ cache-nodes and $\binom{\Gamma}{L}$ users, where each subset of $L$ cache-nodes is connected to a distinct user.
A coded caching scheme for this MACC model was proposed in \cite{MKRmn}, which was shown to be optimal  under the constraint of uncoded cache placement~\cite{FP}.
 Another access topology was considered based on a combinatorial structure  known as cross-resolvable design~\cite{KMR}  where this topology contains  $\Gamma = mq$ cache-nodes serving exactly $\binom{m}{L}q^L$ users, and   each cache has size $M = N/q$ for any positive integer $q$ and $m\geq L$.\footnote{One can check that $\binom{m}{L}q^L=\binom{\Gamma/q}{L}q^L=\mathcal{O}((\Gamma/q)^L q^L)=\mathcal{O}(\Gamma^L).$} 

\paragraph*{Main contributions}
Although numerous works have been proposed in the literature for the MACC model, given the number of cache-nodes $\Gamma$ and the user access degree $L$, the number of users in the system is either linear with $\Gamma$, or of order $\Gamma^L$.
When these schemes are applied in practical systems, a large number of virtual users may be required. The main contribution of this paper is to consider the access topology with a more flexible number of users. By formulating the access topology in terms of two classical combinatorial structures, $t$-design and $t$-group divisible design (also known as $t$-GDD), we propose two classes of coded caching schemes for a
flexible number of users, where the number of users can be scaled linearly, polynomially, or exponentially with the number of cache-nodes. Note that the access topologies in \cite{MKRmn,KMR} can be considered as special cases of ours.

Since the combinatorial structures considered in \cite{JHNS,CWLZC,RKstructure,MKRmn,KMR}, i.e. the cyclic wrap-around access, the combinatorial access and the cross-resolvable design access, are different, the authors proposed topology-dependent cache placement and delivery strategies. In this paper,
we aim to propose unified coded caching strategies for arbitrary parameters of $t$-designs and $t$-GDDs access topology by exploiting the structure of the $t$-design and $t$-GDD, respectively. Our main contributions are as follows.
\begin{itemize}
\item According to the property of $t$-design, we apply the placement strategy of the MN scheme into the $t$-design access topology and delivery strategy generated via the   property of $t$-design (Definition \ref{def-design}) to obtain our first scheme. Interestingly, our scheme covers the schemes in \cite{MKRmn} as special cases.

\item According to the property of $t$-GDD, 
 by applying an Orthogonal Array (OA) placement strategy into  the $t$-GDD access topology and the delivery strategy generated by the   property of $t$-GDD (Definition \ref{def-GDD}), we obtain the second scheme proposed in this paper. Furthermore, we show that the access topology, i.e. the cross-resolvable design, in \cite{KMR} can be considered as a special case of ours. Compared with the schemes in \cite{MR,MKR,KMR}, our scheme has smaller transmission load and subpacketization.
\end{itemize}
As a by-product, the proposed schemes can also work and unify some existing schemes for the original MN shared-link coded caching model, where each user has its own cache. In particular, our $t$-design scheme covers the   schemes in \cite{MN,YTCC}, and our $t$-GDD scheme covers the PDA schemes in \cite{SZG,TR,YCTC,CWZW}, as special cases.

\paragraph*{Notations}  In this paper, the following notations will be used unless otherwise stated.
\begin{itemize}
\item $|\cdot|$ denotes the cardinality of a set.
\item For a set $\mathcal{V}$, we sort it in the lexicographic order; let   $\mathcal{V}(j)$ represent the $j$-th smallest element in   $\mathcal{V}$  and let $\mathcal{V}(\mathcal{J})=\{\mathcal{V}(j)\ |\ j\in\mathcal{J} \}$.
\item For any positive integers $a$, $b$, $t$ with $a<b$ and $t\leq b $,   let $[a:b]=\{a,a+1,\ldots,b\}$,   $[1:b]=[b]$, and
${[b]\choose t}=\{  \mathcal{V}\subseteq [b]\ |\  |\mathcal{V}|=t\}$, i.e., ${[b]\choose t}$ is the collection of all $t$-subsets of $[b]$.
\item Let ${\bf a}$ be a vector with length $n$. For any $i\in[n]$, ${\bf a}(i)$ denotes the $i^{\text{th}}$ coordinate of ${\bf a}$.
For any subset $\mathcal{T}\subseteq [n]$, ${\bf a}(\mathcal{T})$ denotes a vector with length $|\mathcal{T}|$ obtained by taking only the coordinates with indices in $ \mathcal{T}$.
\item For any $F\times m$ array $\mathbf{P}$, for any integers $i\in[F]$ and $j\in [m]$, $\mathbf{P}(i,j)$ represents the element located in the $i^{\text{th}}$ row and the $j^{\text{th}}$ column of $\mathbf{P}$; $\mathbf{P}(\mathcal{ V},\mathcal{T})$ represents the subarray generated by the row indices in $\mathcal{V}\subseteq [F]$ and the columns indices in $\mathcal{T}\subseteq [m]$. In particular let $\mathbf{P}([F],\mathcal{T})$  be shortened by $\mathbf{P}(\cdot,\mathcal{T})$ and  $\mathbf{P}(\mathcal{V},[m])$  be shortened by $\mathbf{P}(\mathcal{V},\cdot)$.
\item If $a$ is not divisible by $q$, $\langle a\rangle _q$ denotes the least non-negative residue of $a$ modulo $q$; otherwise, $\langle a\rangle _q:=q$.
\item For any two vectors ${\bf x}$ and ${\bf y}$ with the same length, $d({\bf x},{\bf y})$, which is called the hamming distance of ${\bf x}$ and ${\bf y}$, i.e., the number of coordinates in which ${\bf x}$ and ${\bf y}$ differ.
\end{itemize}

\section{System Model of MACC}\label{sec-pre}
In this section, we first introduce the original shared-link MN coded caching model in~\cite{MN}, where each user has its own cache (i.e. the number of cache-nodes and the number of users are equal, while each user accesses a different cache-node). We then present the  MACC model considered in this paper, where each user can access an arbitrary subset of cache-nodes.
\subsection{Shared-link coded caching model}
\label{sub:ori model}
In the shared-link coded caching system~\cite{MN}, a server containing $N$ files of equal length in the library $\mathcal{W}=\{W_{1}, W_{2}, \ldots, $ $W_{N}\}$ is connected by an error-free shared link to $K$ users in  $\{U_1,U_2,\ldots,U_K\}$ with $K\leq N$, and each user has a cache of size   $M$ files where $0\leq M \leq N$. The memory ratio is defined as $\mu=M/N$.
An $F$-division $(K,M,N)$ coded caching scheme contains two phases:
\begin{itemize}
  \item Placement phase. Each file is divided into $F$ packets of equal size.\footnote{\label{foot:uncoded}In this paper, we only consider uncoded cache placement.}
   Each user $U_k$, where $k\in [K]$, caches a total of up to $MF$ packets of files. Let $\mathcal{Z}_{U_k}$ denote the cache content  at user $U_k$. Note that the placement phase is done without knowledge of later requests.
  \item Delivery phase. Each user   requests a file from the server. The file requested by user $U_k$ is denoted by $W_{d_{U_k}}$, and the request vector of all users is denoted by $\mathbf{d}=(d_{U_1},d_{U_2},\ldots,d_{U_K})$. According to the users' cache content and requests, the server broadcasts $S_{{\bf d}}$ packets to all users so that each user's request can be satisfied.
\end{itemize}

In such system, the worst-case number of  transmitted files (a.k.a. worst-case load, or simply load) among all possible requests is expected to be as small as possible, which is defined as
\begin{align}
R=\max\left\{S_{\mathbf{d}}/F \ |\  \mathbf{d}\in[N]^K \right\} . \label{eq:def of load}
\end{align}
For the shared-link coded caching problem, the authors in \cite{YCTC} proposed a combinatorial structure to construct  coded caching schemes with uncoded cache placement (symmetric cross files) and clique-covering delivery,\footnote{\label{foot:clique-covering}Clique-covering delivery means that each multicast message sent is the XOR of some subfiles and is useful to a subset of users, where each user requests one subfile and caches all the other subfiles. }  referred to as a Placement Delivery Array (PDA), with the following definition.
\begin{definition}\rm(\cite{YCTC})
\label{def-PDA}
For  positive integers $K,F, Z$ and $S$, an $F\times K$ array  $\mathbf{P}=(\mathbf{P}(j,k))_{j\in[F] ,k\in[K]}$, composed of a specific symbol $``*"$  and $S$ positive integers from $[S]$, is called a $(K,F,Z,S)$  PDA  if it satisfies the following conditions:
\begin{enumerate}
  \item [C$1$.] the symbol $``*"$ appears $Z$ times in each column;
  \item [C$2$.] each integer occurs at least once in the array;
  \item [C$3$.] for any two distinct entries $\mathbf{P}(j_1,k_1)$ and $\mathbf{P}(j_2,k_2)$,    $\mathbf{P}(j_1,k_1)=\mathbf{P}(j_2,k_2)=s$ is an integer  only if
  \begin{enumerate}
     \item [a.] $j_1\ne j_2$, $k_1\ne k_2$, i.e., they lie in distinct rows and distinct columns; and
     \item [b.] $\mathbf{P}(j_1,k_2)=\mathbf{P}(j_2,k_1)=*$, i.e., the corresponding $2\times 2$  subarray formed by rows $j_1,j_2$ and columns $k_1,k_2$ must be of the following form
  \begin{align*}
    \left(\begin{array}{cc}
      s & *\\
      * & s
    \end{array}\right)~\textrm{or}~
    \left(\begin{array}{cc}
      * & s\\
      s & *
    \end{array}\right).
  \end{align*}
   \end{enumerate}
\end{enumerate}
\hfill $\square$
\end{definition}
 In a $(K,F,Z,S)$ PDA $\mathbf{P}$,  for each $k\in [K]$,  column $k$ represents    user $U_k$; for each $j\in [F]$,  row $j$ represents the $j^{th}$ packet of each file. If $\mathbf{P}(j,k)=*$, then user $U_k$ caches the $j^{th}$ packet of all files. If $\mathbf{P}(j,k)=s$ is an integer, then the $j^{th}$ packet of each file  is not cached by user $U_k$.
  In addition, a multicast message, which is  the XOR of the requested packets indicated by $s$,  is sent by the server at time slot $s$. Condition C$2$ of the definition \ref{def-PDA} implies that the number of multicast messages sent by the server is exactly $S$. So the total transmission load is $R=S/F$. Finally, Condition C$3$   of Definition \ref{def-PDA} guarantees that each user can get the packet it requests, since it has cached all the other packets in the multicast message useful to it  except the packet it requests.  Thus, the following lemma was proved by Yan et al. in \cite{YCTC}.
\begin{lemma}\rm(\cite{YCTC})
\label{le-Fundamental}Using Algorithm \ref{alg:PDA}, an $F$-division caching scheme for the $(K,M,N)$ shared-link coded caching model can be realized by a $(K,F,Z,S)$ PDA  with $\frac{M}{N}=\frac{Z}{F}$. Each user can decode its requested file correctly for any request ${\bf d}$ with the load $R=\frac{S}{F}$.
\hfill $\square$
\end{lemma}
\begin{algorithm}[htb]
\caption{Coded caching scheme for Shared-link caching model based on PDA in \cite{YCTC}}\label{alg:PDA}
\begin{algorithmic}[1]
\Procedure {Placement}{$\mathbf{P}$, $\mathcal{W}$}
\State Divide each file $W_n\in\mathcal{W}$ into $F$ packets, i.e., $W_{n}=(W_{n,j}\ |\ j\in [F])$.
\For{$k\in [K]$}
\State $\mathcal{Z}_{U_{k}}\leftarrow\{W_{n,j}\ |\ \mathbf{P}(j,k)=*, \forall~n\in[N]\}$
\EndFor
\EndProcedure
\Procedure{Delivery}{$\mathbf{P}, \mathcal{W},{\bf d}$}
\For{$s=1,2,\cdots,S$}
\State  Server sends $\bigoplus_{\mathbf{P}(j,k)=s,j\in[F],k\in[K]}W_{d_{U_k},j}$.
\EndFor
\EndProcedure
\end{algorithmic}
\end{algorithm}

Let us briefly introduce the MN coded caching scheme in~\cite{MN} from the viewpoint of PDA, where the resulting PDA is called an MN PDA. For any $\mu\in \{\frac{1}{K}, \frac{2}{K},\ldots,\frac{K-1}{K},1\}$ we let $F={K\choose \mu K}$. We sort all $(\mu K+1)$-subsets of $[K]$  in lexicographic order and define $\phi(\mathcal{S})$ as its order for each $(\mu K+1)$-subset $\mathcal{S}$. Clearly, $\phi$ is a bijection from ${[K]\choose \mu K+1}$ to $[{K\choose \mu K+1}]$.
Then, an MN PDA is defined as a ${K\choose \mu K}\times K$ array $\mathbf{P}=(\mathbf{P}(\mathcal{D},k))_{\mathcal{D}\in {[K]\choose \mu K},k\in [K]}$
by
\begin{align}\label{Eqn_Def_AN}
\mathbf{P}(\mathcal{D},k)=\left\{\begin{array}{cc}
\phi(\mathcal{D}\cup\{k\}), & \mbox{if}~k\notin\mathcal{D}\\
*, & \mbox{otherwise}
\end{array}
\right.
\end{align}
where each of the $\binom{K}{\mu K}$ rows in the array is labelled by a distinct  subset  in $ {[K]\choose \mu K}$. Thus, according to the definition of PDA, the achieved load by the MN PDA is as follows.
\begin{lemma}\rm(MN PDA\cite{MN})
\label{le-MN}
For any positive integer $K$ and each $\mu\in \{\frac{1}{K}, \frac{2}{K},\ldots,\frac{K-1}{K},1\}$, there exists a $\left(K,{K\choose \mu K}, {K-1\choose \mu K-1}, {K\choose \mu K+1}\right)$ PDA, which leads a ${K\choose \mu K}$-division $(K,M,N)$ shared-link coded caching scheme with memory ratio $\frac{M}{N}=\frac{t}{K}$ and   load $R=\frac{K-\mu K}{t+1}$.
\hfill $\square$
\end{lemma}

\begin{example}\label{MN-pda}
\rm
We then illustrate the MN PDA by the following example where $N=K=4$ and  $M=2$.
By the construction of the MN PDA, we have  the following   $(4,6,3,4)$ PDA,
\begin{align}
\label{eq-PDA-6-4}
\mathbf{P}=\left(\begin{array}{cccc}
*	&	*	&	1	&	2	\\
*	&	1	&	*	&	3	\\
*	&	2	&	3	&	*	\\
1	&	*	&	*	&	4	\\
2	&	*	&	4	&	*	\\
3	&	4	&	*	&	*	
\end{array}\right).
\end{align}
Using Algorithm \ref{alg:PDA}, the detailed caching scheme is as follows.
\begin{itemize}
   \item \textbf{Placement Phase}: From Line 2 of the algorithm we have $W_n=(W_{n,1},W_{n,2},W_{n,3},W_{n,4},W_{n,5},W_{n,6})$ where  $n\in [4]$. Then by Lines 3-5 in Algorithm \ref{alg:PDA}, the users' caches are
       \begin{align*}
       \mathcal{Z}_{U_1}=\left\{W_{n,1},W_{n,2},W_{n,3}\ |\ n\in[4]\right\};\ \ \ \ \ \ \
       \mathcal{Z}_{U_2}=\left\{W_{n,1},W_{n,4},W_{n,5}\ |\ n\in[4]\right\}; \\
       \mathcal{Z}_{U_3}=\left\{W_{n,2},W_{n,4},W_{n,6}\ |\ n\in[4]\right\};\ \ \ \ \ \ \
       \mathcal{Z}_{U4}=\left\{W_{n,3},W_{n,5},W_{n,6}\ |\ n\in[4]\right\}.
       \end{align*}
   \item \textbf{Delivery Phase}: Assume that the request vector is $\mathbf{d}=(1,2,3,4)$. By the transmitting process by Lines 8-10 in Algorithm~\ref{alg:PDA},  the server transmits the multicast messages in Table~\ref{table1}, with total   load   $R=\frac{4}{6}=\frac{2}{3}$.
   \begin{table}[!htp]
  \normalsize{
  \begin{tabular}{|c|c|}
\hline
   Time slot& Multicast messages  \\
\hline
   $1$&$W_{1,4}\oplus W_{2,2}\oplus W_{3,1}$\\ \hline
   $2$&$W_{1,5}\oplus W_{2,3}\oplus W_{4,1}$\\ \hline
  $3$& $W_{1,6}\oplus W_{3,3}\oplus W_{4,2}$\\ \hline
   $4$& $W_{2,6}\oplus W_{3,5}\oplus W_{4,4}$\\ \hline
  \end{tabular}}\centering
  \caption{Delivery phase in Example \ref{MN-pda} }\label{table1}
\end{table}
\end{itemize}
\hfill $\square$
\end{example}
In the literature of shared-link coded caching model, there are many constructions based on  PDA~\cite{CJWY,CJYT,CJTY,CWZW,ZCJ,ZCW,YCTC,YTCC,SZG,MW}. Some other constructions  could be also transformed into PDAs such as  the caching schemes based on  liner block code~\cite{TR}, projective space~\cite{CKSM}, combinatorial designs~\cite{ASK}, etc.  We list  their performances in Table \ref{knownPDA1}.

\begin{table}
  \centering
  \caption{Existing PDAs where $K,K_1,m,t,q,a,b,r\in\mathbb{Z}^+$, $\mu\in \{i/K\ |\ i\in [0:K]\}$ and  $\mu_1\in \{i/K\ |\ i\in [0:K_1]\}$. Define that $\left[m \atop t\right]_q=\frac{(q^m-1)\cdots(q^{m-t+1}-1)}{(q^t-1)\cdots(q-1)}$. The greatest common divisor of $k$ and $K$ is denoted by $gcd(k,K)$.}  \label{knownPDA1}
  \begin{tabular}{|c|c|c|c|c|}
\hline
PDAs                          &$K$ & $F$& $Z$& $S$ \\
\hline
\cite{MN,YCTC,CWZW} & $K$  & ${K\choose \mu K}$ & ${K-1 \choose \mu K-1}$ & ${K\choose \mu K+1}$\\
 \hline

\cite{SJTLD,CJWY}:  &$K$ &$\frac{K_1}{gcd(K_1,K)} {K_1\choose \mu_1K_1}$ &$\frac{K_1}{gcd(K_1,K)} {K_1-1\choose \mu_1 K_1-1}$& $\frac{K}{gcd(K_1,K)} {K_1\choose \mu_1 K_1+1}$ \\
\hline

\multirow{2}{*}{\tabincell{c}{\cite{YCTC,TR,CWZW}: \\$m\geq 2,q\geq 2$}} &  \multirow{2}{*}{$mq$} & $q^{m-1}$  & $q^{m-2}$ & $(q-1)q^{m-1}$ \\  \cline{3-5}

  & & $(q-1)q^{m-1}$ & $(q-1)^2q^{m-2}$ & $q^{m-1}$  \\
\hline

\tabincell{c}{\cite{SZG,CJYT,CWZW}: \\$b< m$, $q\geq 2$}& ${m\choose t}q^t$ & $q^m$ & $\left(1-\left(\frac{q-1}{q}\right)^t\right)q^m$  &  $(q-1)^tq^m$ \\[8pt]
 \hline

\multirow{2}{*}{\tabincell{c}{\\[-0.2cm] \cite{CJYT}: $z<q$,\\ $t< m$}} &${m\choose t}q^t$ & $\left\lfloor\frac{q-1}{q-z}\right\rfloor^tq^m$ & $\left(1-\left(\frac{q-z}{q}\right)^t\right)\left\lfloor\frac{q-1}{q-z}\right\rfloor^tq^m$ & $(q-z)^tq^m$ \\ [8pt] \cline{2-5}

 & $mq$ & $\lfloor\frac{q-1}{q-z}\rfloor q^{m-1}$ &  $\lfloor\frac{q-1}{q-z}\rfloor q^{m-2}$  & $(q-z)q^{m-1}$ \\ [5pt]

 \hline

\multirow{2}{*}{\tabincell{c}{\\[-0.2cm] \cite{CWZW}: $t< m$,\\ $q\geq 2$}}& ${m\choose t}q^t$ & $q^{m-t}$ &  $\left(1-\left(\frac{q-1}{q}\right)^t\right)q^{m-t}$  & $(q^t-1)q^{m-t}$ \\ [8pt] \cline{2-5}

 &${m\choose t}q^t$ & $q^{m-1}$ & $\left(1-\left(\frac{q-1}{q}\right)^t\right)q^{m-1}$ & $(q-1)^tq^{m-1}$ \\ [8pt]
 \hline

\tabincell{c}{\cite{YTCC,CJTY}: $\eta <a$\\ $<m$, $\eta <b<m$, \\$a+b\leq m+\eta $} &  ${m\choose a}$ & ${m\choose b}$ &  $\left(1-\frac{{a\choose \eta }{m-a\choose b-\eta }}{{m\choose b}}\right){m\choose b}$  & \tabincell{c}{${m\choose a+b-2\eta }\cdot$\\ $\min\{{m-a-b+2\eta \choose \eta },$\\ ${{a+b-2\eta \choose a-\eta }}\}$} \\
\hline

\multirow{2}{*}{\tabincell{c}{\cite{CKSM}: $a+t\leq m$,\\ prime power $q\geq 2$}}& \tabincell{c}{$\frac{1}{t!}q^{\frac{t(t-1)}{2}}$\\$\prod_{i=0}^{t-1}\left[m-i \atop 1\right]_q$}& \tabincell{c}{$\frac{1}{a!}q^{\frac{a(a-1)}{2}}$\\$\prod_{i=0}^{a-1}\left[m-i \atop 1\right]_q$}
&\tabincell{c}{$\frac{1}{a!}q^{\frac{a(a-1)}{2}}\left(\prod_{i=0}^{a-1}\left[m-i \atop 1\right]_q\right.$\\ $\left.-q^{at}\prod_{i=0}^{a-1}\left[m-t-i \atop 1\right]_q\right)$}  &\tabincell{c}{$\frac{1}{(a+t)!}q^{\frac{(a+t)(a+t-1)}{2}}$\\$\prod_{i=0}^{a+t-1}\left[m-i \atop 1\right]_q$} \\
\cline{2-5}
&$\left[m \atop t\right]_q$& $\left[m \atop a+t\right]_q$ &
$\left[m \atop a+t\right]_q-\left[m-t \atop a\right]_q$&
$\left[m \atop a\right]_q$\\
 \hline

\tabincell{c}{\cite{ASK}: prime power \\ $q\geq 2$}&$q^2+q+1$& $q^2+q+1$ &$q^2$&$q^2+q+1$ \\
\hline

\end{tabular}
\end{table}

\subsection{Multiaccess coded caching}
A  multiaccess coded caching problem contains a server with a library of $N$ equal-length files  (denoted by $\mathcal{W}=\{W_{1},  \ldots,   W_{N}\}$), $\Gamma$ cache-nodes (denoted by $ C_1, \ldots,C_{\Gamma} $), and $K\leq N$ users (denoted by $ U_{1},  \ldots, U_{K} $). Each cache-node has a cache of size  $M$ files where $0\leq M \leq  \frac{N}{L}$. The memory ratio is defined as $\mu=M/N$.
For any integer $k\in [K]$, $\mathcal{B}_k$ denotes the set of cache-nodes accessible to user $U_k$. Assume that each user $k\in [K]$ can access a distinct set of $L$ cache-nodes, i.e., $|\mathcal{B}_k|=L$. We call the above caching model as  $(L,K,\Gamma,M,N)$ multiaccess coded caching system with access topology $\mathfrak{B}=\{\mathcal{B}_k\ |\ k\in[K]\}$.

An $F$-division  $(L,K,\Gamma,M,N)$ coded caching scheme with access topology $\mathfrak{B}$  runs in two phases:
\begin{itemize}
\item {\bf Placement phase:} Each file is divided into $F$ packets of equal size. Each cache-node $C_{\gamma}$ where $\gamma\in[\Gamma]$, caches a totally of up to $MF$ packets of   files. Let $\mathcal{Z}_{C_{\gamma}}$ denote the cache content at cache-node $C_{\gamma}$. The placement phase is also done without knowledge of later requests. Each user $U_{k}$ where $k \in [K]$ can retrieve the packets cached at the $L$ cache-nodes  in  $\mathcal{B}_k$. Let $\mathcal{Z}_{U_k}$ denote the packets retrievable by user $U_k$.
\item {\bf Delivery phase:} Each user requests one file. According to the request vector $\mathbf{d}=(d_{U_1},d_{U_2}$, $\ldots$, $d_{U_{K}})$,
     the cache  content in cache-nodes and the  access topology, the server broadcasts $S_{{\bf d}}$ coded packets to all users, such that each user's request can be satisfied.
\end{itemize}

We aim to design multiaccess  coded caching schemes with minimum worst-case load as defined in~\eqref{eq:def of load}.

The PDA construction (review in Section~\ref{sub:ori model}) could be extended to the MACC model, to construct MACC coded caching schemes with
 with the uncoded placement and one-shot delivery, as proposed in \cite{CWLZC}. The placement and delivery strategies of such MACC can be characterised by three arrays defined as follows.
\begin{definition}[\cite{CWLZC}]\rm
\label{defn:three arrays}
\begin{itemize}
\item An $F\times \Gamma$ node-placement array $\mathbf{C}$ consists of star and null, where $F$ and $\Gamma$ represent the subpacketization of each subfile and the number of cache-nodes,  respectively. For any integers $j\in [F]$ and $\gamma\in [\Gamma]$, the entry $\mathbf{C}(j,\gamma)$ is star if and only if the $\gamma^{\text{th}}$ cache-node caches the $j^{\text{th}}$ packet of each $W_n$ where $n\in [N]$.
\item An $F\times K$ user-retrieve array $\mathbf{U}$ consists of star and null, where $F$ and $K$ represent the subpacketization of each subfile and the number of users, respectively. For any integers $j\in [F]$ and $\gamma\in [\Gamma]$, the entry $\mathbf{U}(j,k)$ is a star if and only if the $k^{\text{th}}$ user can retrieve the $j^{\text{th}}$ packet  of each $W_n$ where $n\in [N]$ from its connected cache-nodes.
\item An $F\times K$ user-delivery array $\mathbf{Q}$ consists of $\{*\}\cup[S]$, where $F$, $K$ and the stars in $\mathbf{Q}$ have the same meaning as $F$, $K$ of $\mathbf{U}$ and the stars in $\mathbf{U}$, respectively. Each integer represents one multicast message, and $S$ represents the total number of multicast messages sent in the delivery phase.
    \hfill $\square$
\end{itemize}
\end{definition}
The authors in \cite{CWLZC} showed that if user-delivery array $\mathbf{Q}$ satisfies   Condition C3 in Definition~\ref{def-PDA}, each user can decode its required file by its retrievable cache content and received coded messages from server. Furthermore if each column of $\mathbf{Q}$ has exactly $Z$ stars, then $\mathbf{Q}$ is a $(K,F,Z,S)$ PDA. In this paper we will  construct the above arrays to obtain new MACC schemes. 

For  the sake of simplicity we do not distinguish between   user $U_k$ and its accessible cache-node set $\mathcal{B}_k$ unless otherwise stated. Then any fixed access topology can be represented by an appropriate combinatorial structure. In this paper we will consider   some classical combinatorial structures in combinatorial design theory as the access topologies.  In the following we will introduce some  necessary combinatorial concepts that will be used in our paper.


\section{Combinatorial Designs}
\label{sec-combinatorial-concepts}
\subsection{$t$-design}
\label{sub-design}
\begin{definition}\rm(\cite{Stinson}, Design)
 A design is a pair $(\mathcal{X}, \mathfrak{B})$ such that the following properties are
satisfied:
\begin{itemize}
\item $\mathcal{X}$ is a set of elements (called points), and
\item $\mathfrak{B}$ is a collection of non-empty subsets of $\mathcal{X}$ (called blocks).
\end{itemize}
\hfill $\square$
\end{definition}
\begin{definition}\rm(\cite{CD}, $t$-design)
\label{def-design}
Let $\Gamma$, $K$, $L$, $t$ and $\lambda$ be positive integers. A $t$-$(\Gamma,L,K,\lambda)$ design is a design $(\mathcal{X}, \mathfrak{B})$ where $\mathcal{X}$ has $\Gamma$ points and $\mathfrak{B}$ has $K$ blocks such that the following properties are
satisfied:
\begin{itemize}
\item $|\mathcal{B}|=L$ for any $\mathcal{B}\in \mathfrak{B}$;
\item every $t$-subset of $\mathcal{X}$ is contained in exactly $\lambda$ blocks.
\end{itemize}
\hfill $\square$
\end{definition}
By Definition \ref{def-design},   the number of blocks is
\begin{align}\label{eq-value-K}
K=\frac{\lambda{\Gamma\choose t}}{{L\choose t}}\approx O\left(\Gamma^t\right) \ \ \ \ \  (\Gamma\rightarrow \infty, \Gamma \gg t).
\end{align}Therefore, a $t$-$(\Gamma,L,K,\lambda)$ design is also referred to as a $t$-$(\Gamma,L,\lambda)$ design  in this paper.
 Note that for a $t$-$(\Gamma,L,K,\lambda)$ design, each point appears the same time in all blocks, i.e., $r= KL/\Gamma=\frac{\lambda \binom{\Gamma-1}{t-1}}{\binom{L-1}{t-1}}$.

A  $t$-$(\Gamma,L,\lambda)$ design is also a $t'$-$(\Gamma,L,\lambda_{t'})$ design where $t'\leq t$ and
\begin{align}\label{eq-occurrence}
\lambda_{t'}= \frac{\lambda{\Gamma-t'\choose t-t'}}{ \binom{L-t'}{t-t'}}.
\end{align}
The $2$-$(\Gamma,L,\lambda)$ design is always called $(\Gamma,L,\lambda)$ balanced incomplete block design (in short $(\Gamma,L,\lambda)$ BIBD).

\begin{example}\rm
\label{exam-1-t-design}
1) When $\Gamma=7$ and $L=3$, let $\mathcal{X}=\{1,2,3,4,5,6,7\}$ and
$$\mathfrak{B}=\{\{1,2,4\},\{2,3,5\},\{3,4,6\},\{4,5,7\},\{5,6,1\},\{6,7,2\},\{7,1,3\}\}.$$
Then $(\mathcal{X},\mathfrak{B})$ is a $2$-$(7,3,1)$ design, i.e., $(7,3,1)$ BIBD.

2) When $\Gamma=9$ and $L=3$, let $\mathcal{X}=\{1,2,3,4,5,6,7,8,9\}$ and
\begin{align*}
\mathfrak{B}=&\{\{1,4,7\},\{2,5,8\},\{3,6,9\},\{1,2,3\},\{4,5,6\},\{7,8,9\},\{1,6,8\},
\{2,4,9\},\{3,5,7\}, \\&
\{1,5,9\},\{2,6,7\},\{3,4,8\}
\}.
\end{align*}
Then $(\mathcal{X},\mathfrak{B})$ design is a $2$-$(9,3,1)$ design, i.e., $(9,3,1)$ BIBD.
\hfill $\square$
\end{example}

Recently, the following sufficient and necessary condition of the existence   of  $t$-design was proved in \cite{Peter,GKLO}.
\begin{lemma}[\cite{Peter,GKLO}]\rm
\label{conjecture-design}
Given $t$, $L$ and $\lambda$, there exists an integer $\Gamma_0(t,L,\lambda)$ (which is a function of $(t,L,\lambda)$) such that
for any $\Gamma>\Gamma_0(t,L,\lambda)$, a $t$-$(\Gamma,L,\lambda)$ design exists if and only if for any $0\leq i\leq t-1$,
${L-i\choose t-i}$ divides $\lambda {\Gamma-i\choose t-i}$.
\hfill $\square$
\end{lemma}


For example, if $t=2$ and $L=3$,  we have $\Gamma_0(2,3,1)=\Gamma_0(2,3,2)=2$ in
\cite[Page 72, Section 3, Chapter II]{CD}.
Since the general expression of $\Gamma_0(t,L,\lambda)$ is very complicated, the interested reader could refer to \cite{Peter,GKLO} for the details.

\subsection{$t$-GDD}
\label{sub-GDD}
Next we will review another  specific type of  design, referred to as  group divisible design, which is also useful to our later construction.
\begin{definition}\rm(\cite{CD}, $t$-GDD)
\label{def-GDD}
Let $L$, $t$, $q$ and $m$ be positive integers with $t\leq L\leq m$. A $(m ,q, L,\lambda)$ {\em group divisible $t$-design} (a.k.a, $t$-$(m ,q, L,\lambda)$ GDD) is a triple $(\mathcal{X}, \mathfrak{G}, \mathfrak{B})$ where \begin{itemize}
\item $\mathcal{X}$ is a set of $\Gamma=mq$ points;
\item $\mathfrak{G}=\{\mathcal{G}_1,\mathcal{G}_2,\ldots,\mathcal{G}_{m}\}$ is a partition of $\mathcal{X}$ into $m$ non-overlapping subsets (called groups), where each subset has size $q$;
\item $\mathfrak{B}$ is a family of $L$-blocks of $\mathcal{X}$ such that every block intersects every group in at most one point, and every $t$-subset of points from $t$ distinct groups belongs to exactly $\lambda$ blocks.
\end{itemize}\hfill $\square$
\end{definition}
We can obtain the number of blocks in $\mathfrak{B}$ by Definition \ref{def-GDD} as follows.
\begin{align}\label{eq-GDD-blocks-number}
K=\frac{\lambda{m\choose t}q^t}{{L\choose t}}
\end{align}
From \eqref{eq-value-K} and \eqref{eq-GDD-blocks-number} the number of blocks in a $t$-$(\Gamma,L,\lambda)$ design is no smaller than that of a $t$-$(m ,q, L,\lambda)$ GDD, since
\begin{align*}
\frac{\lambda{\Gamma\choose t}}{{L\choose t}}=\frac{\lambda{mq\choose t}}{{L\choose t}} \geq \frac{\lambda{m\choose t}q^t}{{L\choose t}},
\end{align*}
where the equality holds when $q=1$. In other words, when $q=1$, the $t$-$(m ,1, L,\lambda)$ GDD is exactly $t$-$(m,L,\lambda)$ design.
\begin{remark}[Point representation of GDD]\rm
By Definition \ref{def-GDD}, in a $t$-$(m ,q, L,\lambda)$ GDD for any $u\in [m]$ and $v\in [q]$, we can denote the $v^{\text{th}}$ point of $u^{\text{th}}$ group by a vector $(u,v)$.
\end{remark}
\begin{example}\rm
\label{example-2-GDD}
When $m=3$, $q=2$ and $L=3$, we have $\Gamma=mq=6$. Let \begin{eqnarray}
\begin{split}
\mathcal{X}&=\{(1,1),(1,2),(2,1),(2,2),(3,1),(3,2)\} ,  \\
\mathfrak{G}&=\{\mathcal{G}_1=\{(1,1),(1,2)\},
\mathcal{G}_2=\{(2,1),(2,2)\},
\mathcal{G}_{3}=\{(3,1),(3,2)\}\},&\\
\mathfrak{B}&=\{\{(1,1),(2,1),(3,1)\},
\{(1,2),(2,1),(3,2)\},
\{(1,1),(2,2),(3,2)\},
\{(1,2),(2,2),(3,1)\}
\}.&\label{eq-example-3-2-3-block-GDD}
\end{split}
\end{eqnarray}  By Definition \ref{def-GDD} we can check that the triple $(\mathcal{X},\mathfrak{G},\mathfrak{B})$ is a 
$2$-$(3,2,3,1)$ GDD.
\end{example}

By Definition \ref{def-design} and Definition \ref{def-GDD}, we can see that any $t$-design could be seen as a $t$-GDD with $q=1$.
There are many  results on the construction and existence   of the $t$-GDDs;  please refer to \cite[Section IV-4]{CD} for   the details. 
 The existence of the $t$-GDD with $\lambda=1$ is given in the following lemma.
\begin{lemma}(\cite{Mohacsy-2})\rm
 \label{lemma-existence-GDD}
Given $t$, $L$ and $m$ where $t\leq L\leq m$, there exists an integer $q_0(t,L,m)$ (which is a function of $(t,L,m)$) such that
for any $q\geq q_0$, a $t$-$(m,q,L,1)$ GDD exists if and only if for any $0\leq i\leq t-1$,
${L-i\choose t-i}$ divides $q^{t-i}{m-i\choose t-i}$.
\end{lemma}
In the following we will show that the concept of cross resolvable design proposed in \cite{KMR} is a special case of  $t$-GDD.
\begin{definition}\rm(\cite{KMR}, Resolvable design and  cross resolvable design)
\label{def-cross-resolvable-desgin}
A design $(\mathcal{V},\mathfrak{A})$ is called resolvable if the blocks in $\mathfrak{A}$ can be divided into parallel classes where the blocks in each class partition  the set of elements $\mathcal{V}$.
A resolvable design $(\mathcal{V},\mathfrak{A})$ is called  $t$-cross if the intersection of any $t$ blocks drawn from any $t$ distinct parallel classes has the same size $\lambda_t$, referred to as the $t^{\text{th}}$ cross intersection number.\hfill $\square$
\end{definition}

 For the sake of clarification, in the following  we use the $t$-$(v,k,\Gamma,m,\lambda_t)$ resolvable design to represent the $t$-cross resolvable design $(\mathcal{V},\mathfrak{A})$, by letting   $|\mathcal{V}|=v$ and $\mathfrak{A}$ contain $\Gamma$ blocks with of size $k$, where $\mathfrak{A}=\mathfrak{A}_1\cup\mathfrak{A}_2\cup \cdots \cup \mathfrak{A}_m$ can be divided into     $m$ parallel classes.

 For any design $(\mathcal{V},\mathfrak{A})$,  we define its dual design as follows:
 we regard the blocks in $\mathfrak{A}$ as points and the points in $\mathcal{V}$ as blocks, where in the dual design  each point represented by $\mathcal{A}\in \mathfrak{A}$ is contained by a block represented by $a\in \mathcal{V} $ if and only if $a\in \mathcal{A}$;  then the resulting design $(\mathcal{X},\mathfrak{B})$
 is called the dual design of $(\mathcal{V},\mathfrak{A})$. Clearly a design is a dual design of its dual design.

\begin{example}\rm
\label{exam-dual-8-4}
Let $\mathcal{V}=[4]$ and $\mathfrak{A}=\mathfrak{A}_1\cup\mathfrak{A}_2\cup\mathfrak{A}_3$ where
\begin{align}
&\mathfrak{A}_1=\{\mathcal{A}_1=\{1,3\},\mathcal{A}_2=\{2,4\}\}, \ \ \
\mathfrak{A}_2=\{\mathcal{A}_3=\{1,2\},\mathcal{A}_4=\{3,4\}\}, \nonumber \\
&\mathfrak{A}_3=\{\mathcal{A}_5=\{1,4\},\mathcal{A}_6=\{2,3\}\}.\label{eq-cross-design}
\end{align}
It can be seen that the intersection of any two blocks from different parallel classes contains exactly one point. So by Definition \ref{def-cross-resolvable-desgin}, $(\mathcal{V},\mathfrak{A})$ is a $2$-$(v,k,\Gamma,m,\lambda_t)=(4,2,6,3,1)$ resolvable design. Then its dual design is  $(\mathcal{X},\mathfrak{B})$ where
\begin{align}\label{eq-3-2-3-point-GDD}
\mathcal{X}=\mathfrak{A}=\{(1,1)=\mathcal{A}_1,(1,2)=\mathcal{A}_2,
(2,1)=\mathcal{A}_3,(2,2)=\mathcal{A}_4,(3,1)=\mathcal{A}_5,\mathcal{A}_6\}
\end{align} and all the blocks are
\begin{align}
\mathfrak{B}=\{\ &
\{(1,1)=\mathcal{A}_1,(2,1)=\mathcal{A}_3,(3,1)=\mathcal{A}_5\},\nonumber\\
&\{(1,2)=\mathcal{A}_2,(2,1)=\mathcal{A}_3,(2,2)=\mathcal{A}_6\},\nonumber\\
&\{(1,1)=\mathcal{A}_1,(2,2)=\mathcal{A}_4,(3,2)=\mathcal{A}_6\},\nonumber\\
&\{(1,2)=\mathcal{A}_2,(2,2)=\mathcal{A}_4,(3,1)=\mathcal{A}_5\}
\ \}.\label{eq-3-2-3-block-GDD}
\end{align}Let $\mathcal{G}_i=\mathfrak{A}_i$ for each $i\in [3]$, then we can define the set of groups
\begin{align}
\mathfrak{G}=\{\ &\mathcal{G}_1=\{(1,1)=\mathcal{A}_1,
(1,2)=\mathcal{A}_2\},\nonumber\\
&\mathcal{G}_2=\{(2,1)=\mathcal{A}_3,
(2,2)=\mathcal{A}_4\},\nonumber\\
&\mathcal{G}_{3}=\{(3,1)=\mathcal{A}_5,(3,2)=\mathcal{A}_6\}\ \}.
\label{eq-3-2-3-group-GDD}
\end{align}
We can check that $(\mathcal{X},\mathfrak{G},\mathfrak{B})$ is a $t$-$(m,q,L,\lambda)=$ $2$-$(3,2,3,1)$ GDD  which is exactly the GDD in Example \ref{example-2-GDD}.
\hfill $\square$
\end{example}
Using the above transformation method in Example \ref{exam-dual-8-4}, we can obtain the following result that the dual of  any cross resolvable design is a GDD, whose proof could be found in Appendix~\ref{sec:duality}.
\begin{lemma}\rm
\label{le-cross-GDD}
The dual of a $t$-$(v,k,\Gamma,m,\lambda_t)$ resolvable design $(\mathcal{V}, \mathfrak{A})$ is a $t$-$(m,q,m,\lambda_t)$ GDD where $q=v/k$.
\hfill $\square$
\end{lemma}
\subsection{Orthogonal array (OA)}
\label{sub-OA}
Next, we review a classical combinatorial structure called an orthogonal array, which will be used in the placement phase of the proposed caching scheme for the $t$-GDD access topology.
\begin{definition}\rm(\cite{Stinson}, OA)
\label{def-OA}
Let $\mathbf{A}$ be an $F_1\times m$ matrix whose element is in $[q]$,  for positive integers $F_1$, $m$,  $q\geq 2$, and $s\leq m$. $\mathbf{A}$ is an orthogonal array (OA) of strength $s$, denoted by OA$_{\lambda}(F_1,m,q,s)$, if each $1\times s$ row vector  in $[q]^s$   appears exactly $\lambda$ times in  $\mathbf{A}(\cdot,\mathcal{S})$ for each $\mathcal{S}\in {[m]\choose s}$.
\hfill $\square$
\end{definition}
By definition, we have $F_1=\lambda q^s$;  then any OA$_{\lambda}(F_1,m,q,s)$ can  be also written as OA$_{\lambda}(m,q, s)$ for short \cite{Stinson}. The parameter $\lambda$ is the index of the orthogonal array. If $\lambda$ is omitted, then it is understood to be $1$.
\begin{example}\rm
\label{exam-OA-3-2-3}
Let us consider the following array,
\begin{align}
\mathbf{A}=\left(\begin{array}{ccc}
1&1&1\\
2&1&2\\
1&2&2\\
2&2&1
\end{array}\right).\label{eq-OA-4}
\end{align}
Then we have
\begin{align*}
& \mathbf{A}(\cdot,\{1,2\})=\left(\begin{array}{cc}
1&1\\
2&1\\
1&2\\
2&2
\end{array}\right),\ \
 \mathbf{A}(\cdot, \{1,3\})=\left(\begin{array}{cc}
1&1\\
2&2\\
1&2\\
2&1
\end{array}\right),\ \
 \mathbf{A}(\cdot,\{2,3\})=\left(\begin{array}{cc}
1&1\\
1&2\\
2&2\\
2&1
\end{array}\right).
\end{align*}
Clearly for each $\mathcal{S}\in \{\{1,2\}, \{1,3\} , \{2,3\}\}$,
  each vector in $[2]^2=\{(1,1),(1,2),(2,1),(2,2)\}$   appears exactly once in the rows of   $\mathbf{A}(\cdot, \mathcal{S})$.  So
$\mathbf{A}$ in \eqref{eq-OA-4} satisfies Definition \ref{def-OA}, and thus is  an OA$(3,2,2)$.
\hfill $\square$
\end{example}


OA has been widely studied in combinatorial theory, graph theory and coding theory \cite{CD}.
There are many  results on the construction and existence   of the OAs;  please refer to~\cite[Section II-6,7]{CD} for more details.  In particular, it is well known that for any positive integers $m$, $s$ and $s<m$, there exists an OA$(m,q,s)$ for some prime $q$.  In fact, there is a one-to-one mapping between cross resolvable design and OA. Let us first take an example to show this mapping.
\begin{example}\rm
\label{example-cross-OA}
Recall that the $(\mathcal{V},\mathfrak{A})$ is a $t$-$(v,k,\Gamma,m,\lambda)=2$-$(4,2,6,3,1)$ resolvable design in Example \ref{exam-dual-8-4}.
The dual design of this resolvable design is a
  $t$-$(m,q,m,\lambda)=$ $2$-$(3,2,3,1)$ GDD $(\mathcal{X},\mathfrak{G},\mathfrak{B})$, where $\mathcal{X}=\mathfrak{A}$ in \eqref{eq-3-2-3-point-GDD}, all the groups of $\mathfrak{G}$ are in \eqref{eq-3-2-3-group-GDD} and all the blocks of $\mathfrak{B}$ are in \eqref{eq-3-2-3-block-GDD}. We sort all the points in   $\mathcal{X}$ in an arbitrary fixed order; for example
 we sort all points in $\mathfrak{A}$  in the lexicographic order, i.e.,   $\mathfrak{A}=\{(1,1),(1,2), (2,1),(2,2),(3,1),(3,2)\}$.  
Then we represent each block by a row vector according to the above order; for example, the first block of $\mathfrak{B}$ in \eqref{eq-3-2-3-block-GDD},   $\{(1,1),(2,1),(3,1)\}$ is
represented by a row vector $((1,1),(2,1),(3,1))$.  Similarly, $\mathfrak{B}$ can be represented by the following array,
\begin{align*}
\mathbf{A}_{\mathfrak{B}}=\left(
  \begin{array}{ccc}
(1,1)&(2,1)&(3,1)\\
(1,2)&(2,1)&(3,2)\\
(1,1)&(2,2)&(3,2)\\
(1,2)&(2,2)&(3,1)
  \end{array}
\right).
\end{align*}
Next we replace each element in $\mathbf{A}_{\mathfrak{B}}$   by its second coordinate. For example, $(1,1)$ is replaced by $1$; $(1,2)$ is replaced its second coordinate $2$; $(3,2)$ is replaced by its second coordinate $2$.  Similarly, the resulting array after the above replacement process is
\begin{align*}
\mathbf{A}=\left(
  \begin{array}{ccc}
1&1&1\\
2&1&2\\
1&2&2\\
2&2&1
  \end{array}
\right).
\end{align*}
It is interesting to see that the resulting array $\mathbf{A}$ is an  OA$_{\lambda}(m,v/k,t)=$OA$(3,2,2)$ in \eqref{eq-OA-4}.

Furthermore,  given the  OA in \eqref{eq-OA-4}, we can also obtain the $2$-$(v,k,\Gamma,m,\lambda)=(4,2,6,3,1)$ resolvable design in Example \ref{exam-dual-8-4} by performing the  corresponding inverse operations. Hence, there is a one-to-one mapping between OA and cross  resolvable design.
\hfill $\square$
\end{example}

Using the above transformation method in Example \ref{example-cross-OA}, we have the one-to-one mapping between a cross resolvable design and an OA, whose proof could be found in Appendix~\ref{sec:proof of OA lemma}.
\begin{lemma}\rm
\label{le-cross-OA}
 There is a one-to-one mapping between a $t$-$(v,k,\Gamma,m,\lambda_t)$ resolvable design and an OA$_{\lambda_t}(m,q,t)$ where $q=v/k$.  
\hfill $\square$
\end{lemma}


\section{Main results}
\label{sec-main-results}
We first consider the access topology as a $t$-design. By leveraging
  the MN placement strategy and the properties of  $t$-design, we propose a MACC scheme with the following load, whose description could be found in Section \ref{sec-packing-design}.
\begin{theorem}[$t$-design scheme]\rm
\label{th-PDA-SS}
Consider a $(L,K,\Gamma,M,N)$ multiaccess coded caching system with access topology $\mathfrak{B}$ where $([\Gamma],\mathfrak{B})$ is
  a $t$-$(\Gamma,L,1)$ design, for   positive integers $\Gamma$, $L$ and $t\geq 2$, the lower convex envelop of the following memory-load tradeoff points is achievable,
\begin{align}
\label{eq-load}
(M,R_{\text{Th1}})=\left(\mu N,   \frac{{\Gamma \choose t+\mu\Gamma}-K{L\choose t+\mu\Gamma}-\sum^{\mu\Gamma-1}_{i=1}{L\choose t+i}{\Gamma-L\choose \mu\Gamma -i}}{{\Gamma\choose \mu\Gamma}{L\choose t}} \right),
\end{align}
for each $\mu=0, \frac{1}{\Gamma}, \frac{2}{\Gamma},\ldots, \frac{\Gamma-L}{\Gamma}$.
\hfill $\square$
\end{theorem}
Note that to achieve  each memory-load tradeoff point in~\eqref{eq-load}, we first construct  a MACC scheme with
 whose user-delivery array is a $(K={\Gamma\choose t}/{L\choose t}$, $F={\Gamma\choose \mu\Gamma}{L\choose t}$, $Z=\left({\Gamma\choose \mu\Gamma}-{\Gamma -L\choose \mu\Gamma}\right){L\choose t}$, $S={\Gamma \choose t+\mu\Gamma}$ $-K{L\choose t+\mu\Gamma})$ PDA. Then by observing
each user can re-construct some multicast messages from its retrieval cache content, we use a $[n,k]$  maximum distance separable (MDS) code to further reduce the number of transmissions where
$n={\Gamma \choose t+\mu\Gamma}-K{L\choose t+\mu\Gamma}$ and $k={\Gamma \choose t+\mu\Gamma}-K{L\choose t+\mu\Gamma}-\sum^{\mu\Gamma-1}_{i=1}{L\choose t+i}{\Gamma-L\choose \mu\Gamma -i} $

.


 By Definition \ref{def-design} and Definition \ref{def-GDD}, a $t$-$(m,1,L,\lambda)$ GDD is also a $t$-$(m,L,\lambda)$ design. Next we generalize the access topology in Theorem~\ref{th-PDA-SS} (i.e.,  the $t$-design access topology) to the $t$-GDD access topology and propose a MACC scheme which uses the OA structure into the cache placement by extending the construction of the PDA in Theorem \ref{th-PDA-SS}. The achieved load is stated in the following theorem  and  the description on the scheme could be found in  Section \ref{sec-Scheme-GDD}.
\begin{theorem}[$t$-GDD scheme]\rm
\label{th-GDD}
Consider a $(L,K,\Gamma,M,N)$ multiaccess coded caching system with a $t$-$(m,q,L,1)$ GDD access topology
$\mathfrak{B}$ where $([\Gamma=mq],\mathfrak{G}, \mathfrak{B})$ is
for positive integers $m$, $q$, $L$, $t$ and $s$ with $1\leq t\leq L\leq s\leq m$, if there exists an OA$(m,q,s)$,
 the following memory-load tradeoff point  is achievable,
\begin{align}
  (M, R_{\text{Th2}})=\left(\frac{\left(q^s-(q-1)^Lq^{s-L}\right){L\choose t}}{q^s{L\choose t}} N,    \frac{q^{m-s}(q-1)^t}{{L\choose t}} \right).
\label{eq-load for t GDD}
\end{align}
\hfill $\square$
\end{theorem}

Note that we consider $t$-$(\Gamma,L,\lambda=1)$ design and $t$-$(m,q,L,\lambda=1)$ GDD  as access topologies.
 For $\lambda>1$, we can also extend the proposed schemes to
   the $t$-$(\Gamma,L,\lambda)$ design and $t$-$(m,q,L,\lambda)$ GDD access topologies. For the details, please referred to Section \ref{section-lambda>1-th1}.
\subsection{Performance analysis on the proposed schemes in Theorems~\ref{th-PDA-SS} and~\ref{th-GDD}}
\label{sub-performance-first}
By Lemma \ref{conjecture-design} and Theorem \ref{th-PDA-SS}, we can obtain arbitrary $t$-$(\Gamma,L,\lambda)$ designs for any parameters $t$, $\Gamma$, $L$ and $\lambda$ when $\Gamma$ is larger than a threshold $\Gamma_0$.\footnote{For instance when $t=2$ and $L=3$, it is well known that $\Gamma_0=3$ for any $\lambda\in [3]$.}    Then by Theorem \ref{th-PDA-SS}, we can obtain a $(L,K={\Gamma\choose t}/{L\choose t},\Gamma, M,N)$ multiaccess coded caching scheme with memory ratio $\mu=M/N$, subpacketization $F={\Gamma\choose \mu\Gamma}{L\choose t}$ and transmission load in \eqref{eq-load}.
Note that the number of serving users is $K={\Gamma\choose t}/{L\choose t}$. By letting $t$ vary  from $1$ to $L$, we can have a   wide range of  user numbers,  which can scale   linearly,  polynomially, or exponentially with the number of cache-nodes $\Gamma$. As a result, we can serve a very flexible number of users.  By Lemma \ref{lemma-existence-GDD} and Theorem \ref{th-GDD}, we can also server a very flexible number of users.

For any positive integers $\Gamma$ and $L$,  there always exists a $L$-$\left(\Gamma,L,K={\Gamma\choose L},r={\Gamma-1\choose L-1},1\right)$ design $\left(\mathcal{X}=[\Gamma], \mathfrak{B}={[\Gamma]\choose L}\right)$. In fact, the combinatorial access topology proposed in \cite{MKRmn} is exactly the $L$-$\left(\Gamma,L,K={\Gamma\choose L},r={\Gamma-1\choose L-1},1\right)$ design access topology. Then by Theorem \ref{th-PDA-SS}, we have a $\left({\Gamma\choose L}, {\Gamma\choose \mu\Gamma}, {\Gamma\choose \mu\Gamma}-{\Gamma-L\choose \mu\Gamma},{\Gamma\choose \mu\Gamma+L}\right)$ PDA which leads to a $(L,r={\Gamma-1\choose L-1},K={\Gamma\choose L}, M,N)$ multiaccess coded caching scheme based on  $\left(\mathcal{X}=[\Gamma], \mathfrak{B}={[\Gamma]\choose L}\right)$, i.e., combinatorial access topology, with memory ratio  $\mu=\frac{M}{N}$, subpacketization $F={\Gamma\choose \mu\Gamma}$ and load $R_{\text{Th1}}={\Gamma\choose \mu\Gamma+L}/ {\Gamma\choose \mu\Gamma}$. This scheme is exactly the scheme proposed in \cite{MKRmn}, which was shown to be optimal under uncoded cache placement and combinatorial access topology~\cite{FP}.

Next, let us compare the scheme  in Theorem \ref{th-GDD} with the schemes in \cite{MR,MKR,KMR} for the multiaccess access topology. The scheme  in \cite{KMR} has a  coded multicasting gain equal to $4$  under a $2$-$(q^m,q^{m-1},\frac{q(q^m-1)}{q-1},\frac{q^m-1}{q-1},q^{m-2})$  resolvable design with $L=2$ and memory ratio $M/N=1/q$ where $q$ is any prime power and $m$ is any positive integer larger than $1$.
 In this paper we only devote to studying the coded caching gain that depends on the value of the parameters $L$, $\Gamma$ and memory ratio $M/N=1/q$ where $q$ is any positive integer. Furthermore when the memory ratio is $M/N=1/q$, the scheme in \cite{KMR} is reduced to the scheme in \cite{MR}. So we only need to compare with the scheme with memory ratio $M/N=1/q$ in \cite{MR}, where the authors used the  OA$(m,q,m)$ for any positive integer $m$ and $q\geq 2$ as the placement strategy to obtain the following result.
\begin{lemma}\rm(CRS in \cite{MR} generated by OA$(m,q,m)$)
\label{eq-MKR}
For any positive integer $m$ and $q$, there exists a $(L,K,\Gamma,M,N)
    =\left(t,{m\choose t}q^t,mq,M,N\right)$ multiaccess coded caching for the $t$-$(v,k,\Gamma,m,\lambda)=t$-$(q^m$, $q^{m-1},mq,m,q^{m-t})$ resolvable design access topology  generated by OA$(m,q,m)$ with memory ratio $M/N=1/q$,   transmission load
\begin{align}\label{eq-KMR-load}
R_{\text{MR}}=\frac{\mu_t{q\choose 2}^t{m\choose t}}{v}=
\frac{q^{m-t}{q\choose 2}^t{m\choose t}}{4n}={m\choose t}\left(\frac{q-1}{2}\right)^t
\end{align} and coded gain $g_{\text{MR}}=2^t$.
\hfill $\square$
\end{lemma}
For  the same access topology in Lemma~\ref{eq-MKR}, the proposed scheme in Theorem \ref{th-GDD} is a   $(t$, ${m\choose t}q^t$, $mq$, $M,N)$ multiaccess coded caching scheme with memory ratio $M/N=1/q$,  transmission load $R_{\text{Th2}}=(q-1)^t$ and subpacketization  $q^{m-1}$. We can check that $R_{\text{Th2}}/R_{\text{MR}}=2^t/{m\choose t}<1$ always holds when $t\leq m/2$. So, when $t\leq m/2$, the proposed scheme in Theorem \ref{th-GDD} (with $L=t$ and $s= m -1$)
has lower transmission load  and subpacketization (i.e., $1/q$ times smaller) than that of the scheme in \cite{MR}.

\subsection{Extension of the proposed schemes to the shared-link model}
\label{sub:extension to shared-link}
As a by-product, we can directly extend the proposed schemes in Theorem~\ref{th-PDA-SS} and Theorem~\ref{th-GDD} for the MACC model to the shared-link coded caching model in Section~\ref{sub:ori model} with the same number of users as the MACC model; this is because we can let
 the retrievable content by each user in the MACC model be the the cache content of one user in the shared-link model, while the delivery phase does not change.  Then we can obtain the following results.
\begin{corollary}(PDA via Theorem \ref{th-PDA-SS})\rm
\label{corollary-design}
Given a $t$-$(\Gamma,L,1)$ design, there exists a $(K={\Gamma\choose t}/{L\choose t}$, $F={\Gamma\choose \mu\Gamma}{L\choose t}$, $Z=\left({\Gamma\choose \mu\Gamma}-{\Gamma -L\choose \mu\Gamma}\right){L\choose t}$, $S\leq{\Gamma \choose t+\mu\Gamma}$ $-K{L\choose t+\mu\Gamma})$ PDA. Then the following memory-load tradeoff for shared-link model is achievable
\begin{align}\label{eq-load-PDA-th-1}
(M,R_{\text{Th}_1\text{-Shared-link}})= \left(N-\frac{{\Gamma -L\choose \mu\Gamma}N}{{\Gamma\choose \mu\Gamma}}, \frac{{\Gamma-\mu\Gamma\choose t}}{{t+\mu\Gamma\choose t}{L\choose t}}-\frac{{\Gamma\choose t}{L\choose t+\mu\Gamma}}{{\Gamma\choose \mu\Gamma}{L\choose t}^2}\right).
\end{align}
with subpacketization $F={\Gamma\choose \mu\Gamma}{L\choose t}$. 
\hfill $\square$
\end{corollary}

\begin{corollary}[PDAs via Theorem \ref{th-GDD}]\rm
\label{corollary-GDD}
Given a $t$-$(m,q,L,1)$ GDD, we have  a $(K={m\choose t}q^t/{L\choose t}$, $F= q^s{L\choose t}$,  $Z=\left(q^s-(q-1)^Lq^{s-L}\right){L\choose t}$, $S)$ PDA, where $S\leq (q-1)^tq^{m}$. Then the following memory-load tradeoff for shared-link model is achievable
\begin{align}\label{eq-load-PDA-th-2}
(M,R _{\text{Th}_2\text{-Shared-link}})=
\left(N-\frac{(q-1)^LN}{q^L},
\frac{(q-1)^tq^{m-s}}{{L\choose t}}\right).
\end{align} 
 with subpacketization $F=q^s{L\choose t}$. In addition, we can also compute the exact value of $S$ in the following cases:
\begin{itemize}
\item   $L=t$ and $s=m-1$. We have $S=(q-1)^tq^{m-1}$ and $R_{\text{Th}_2\ \text{Shared-link}}=(q-1)^t$;
\item   $L=t$ and $s+t=m$ with $s>t$. We have $S=(q^t-1)q^{m-t}$ and  $R_{\text{Th}_2\ \text{Shared-link}}=q^t-1$;
\item   $s=m$. We have $S\leq (q-1)^tq^m$ and $R_{\text{Th}_2\ \text{Shared-link}}\leq (q-1)^t/{L\choose t}$.
\end{itemize}
\hfill $\square$
\end{corollary}

Let us then compare our proposed PDAs in Corollary \ref{corollary-design} and Corollary \ref{corollary-GDD}
and the existing PDAs in \cite{MN,YTCC,SZG,TR,YCTC,CWZW} for the shared-link model.

First let us consider the comparison between our PDA in Corollary \ref{corollary-design} and the existing PDAs in  \cite{MN,YTCC}. For any positive integers $\Gamma$ and $L$, there always exists a $L$-$\left(\Gamma,L,K={\Gamma\choose L},1\right)$ design and a $1$-$(\Gamma,1,\Gamma,1)$ design. Thus by Corollary \ref{corollary-design}, we can obtain a $({\Gamma\choose L}, {\Gamma\choose \mu\Gamma}, {\Gamma\choose \mu\Gamma}-{\Gamma-L\choose \mu\Gamma},{\Gamma\choose \mu\Gamma+L})$ PDA which is exactly the PDA with the parameter $\eta=0$, $a=L$, $b=\Gamma\mu$ in \cite{YTCC}, and also a $(\Gamma, {\Gamma\choose \mu\Gamma}, {\Gamma-1\choose \mu\Gamma-1},{\Gamma\choose \mu\Gamma+1})$ PDA which is exact the MN PDA in \cite{MN}.

Now let us compare our PDA in Corollary \ref{corollary-GDD}  with the existing PDAs in \cite{MN,SZG,TR,YCTC,CWZW}. Note that the construction of the PDAs in \cite{CWZW} unifies all the constructions in  \cite{MN,SZG,YCTC,TR}. So in the following we need only to dcompare the PDAs in Corollary \ref{corollary-GDD} and the PDAs in \cite{CWZW}, as stated in the following:
\begin{itemize}
\item when $L=t$, for any positive integers $s+t\leq m$ we have $({m\choose t}q^t, q^s, \left(q^s-(q-1)^Lq^{s-t}\right),S)$ PDAs in Corollary \ref{corollary-GDD} with $S\leq q^m$ which are exactly the last three PDAs in \cite{CWZW} by Table \ref{knownPDA1}. This implies that  the last three PDAs in \cite{CWZW} listed in Table \ref{knownPDA1} can be covered by the PDA in Corollary \ref{corollary-GDD} by letting $L=t$, and can be directly used for the multiaccess network. So we propose a unified construction framework from the viewpoint of multiaccess setting;
\item when $L>t$, we have some other new PDAs in Theorem \ref{th-GDD}. So by Theorem \ref{th-GDD} we have some new schemes for the multiaccess network and the shared-link network.
\end{itemize}

From the above introduction, the relationship between the PDAs in \cite{MN,YTCC,SZG,TR,YCTC,CWZW} and the PDAs in Theorem \ref{th-PDA-SS}, Theorem \ref{th-GDD} can be represented in Fig. \ref{fig-GDD-existing}.
 \begin{figure}[http!]
\centerline{\includegraphics[scale=1]{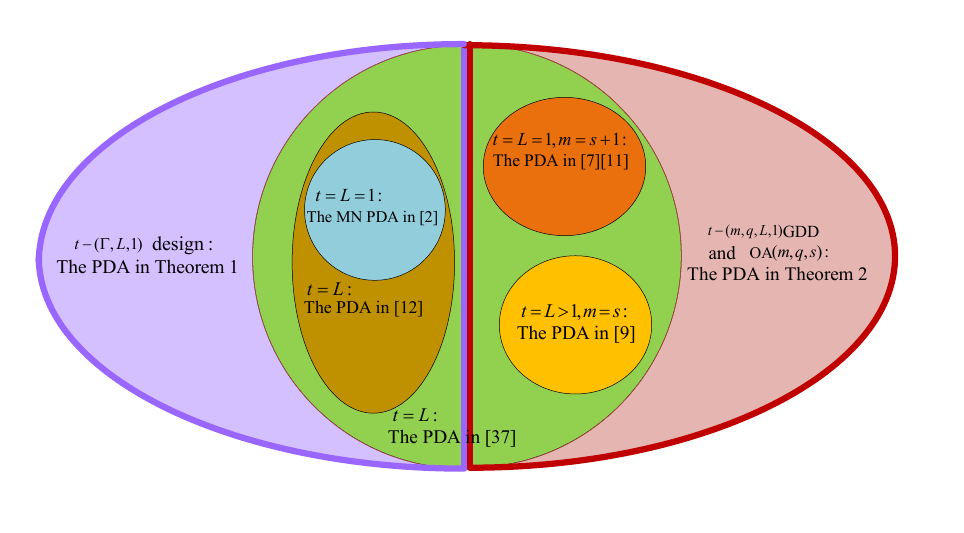}}
\caption{\small The relationship among the PDAs in \cite{MN,YCTC,CWZW,TR,SZG,YTCC} and Corollaries \ref{corollary-design}, \ref{corollary-GDD}.}
\label{fig-GDD-existing}
\end{figure}

Finally, let us compare numerically  our proposed PDAs in Corollary \ref{corollary-design} and Corollary \ref{corollary-GDD}
and the existing PDAs in \cite{YCTC,CWZW,SJTLD,CJWY,TR,SZG,CJYT,CKSM,ASK,CJTY,ZCJ,ZCW,MW} for the shared-link model. For the comparison, it is worth remarking that
\begin{itemize}
\item  all the schemes in \cite{MN,YCTC,TR,SZG} are the special cases of the scheme in \cite{CWZW};
\item we cannot find out an appropriate number of users $K$ to compare each of the schemes in \cite{YTCC,CJWY,CKSM,ZCW,MW,ASK} since all of them work for very special numbers of users;
\item the schemes in \cite{CJTY,ZCW} only work for the specific number of users and the size of  memory such as the  combinatorial numbers;
\item the scheme in \cite{ASK} has the memory ratio approximating to $1$;
\item \cite{ZCJ,CJWY,SJTLD,MW} devote to proposing the methods such that based on a given based PDA, we can obtain a new PDA with flexible user number and subpacketization.
\end{itemize}
From the above remarks, we only need to compare the schemes in \cite{CJYT,YTCC,CKSM}. Let us consider the scheme via $t$-$(\Gamma,L,1)=2$-$(\Gamma,3,1)$ design in Corollary \ref{corollary-design}, the scheme via $t$-$(m,q,L,1)=2$-$(m,q,3,1)$ GDD and a trivial OA$(m,q,m)$ in Corollary \ref{corollary-GDD}.

First let us consider the PDA in Corollary \ref{corollary-design}. Based on a $2$-$(\Gamma,3,1)$ design, the obtained user-delivery array is a $(K,F,Z,S)$ PDA with
\begin{align}
 K=\frac{\Gamma(\Gamma-1)}{6},
 F=3{\Gamma\choose \mu\Gamma}, Z=3\left({\Gamma\choose \mu\Gamma}-{\Gamma-3\choose \mu\Gamma}\right),
 S={\Gamma\choose 2+\mu\Gamma}-K{3\choose 2+\mu\Gamma}.\label{eq-compar-numberic}
\end{align}
By Lemma \ref{le-Fundamental} we have a $(K,M,N)$ coded caching scheme for shared-link network where
\begin{align*}
\frac{M}{N}=1-(1-\mu)\left((1-\frac{\mu\Gamma}{\Gamma-1}\right)\left(1-\frac{\mu\Gamma}{\Gamma-2}\right),\ \ R=\frac{(\Gamma-\mu\Gamma)(\Gamma-\mu\Gamma-1)}{3(2+\mu\Gamma)(1+\mu\Gamma)}.
\end{align*}

When $K=35$, by choosing the parameters $\Gamma=15$ in \eqref{eq-compar-numberic}, the parameter $K=35$ of the MN scheme in \cite{MN}, the parameters $m=7$ and $a=5$ of the scheme in \cite{YTCC}, the parameters $q=2$, $m=4$, $t=2$ and $a=1$ of the scheme in \cite{CKSM}, and the parameters $q=7$ and $m=5$ of the scheme in \cite{CJYT}, we have the transmission loads and the subpacketizations of the schemes in \cite{CJYT,YTCC,CKSM} and Theorem \ref{th-PDA-SS} illustrated in Fig. \ref{fig-design-R} and Fig. \ref{fig-design-F}, respectively. It can be seen the proposed PDAs have the advantages on subpacketizations or transmission loads.
\begin{figure}[http!]
\centerline{\includegraphics[scale=0.7]{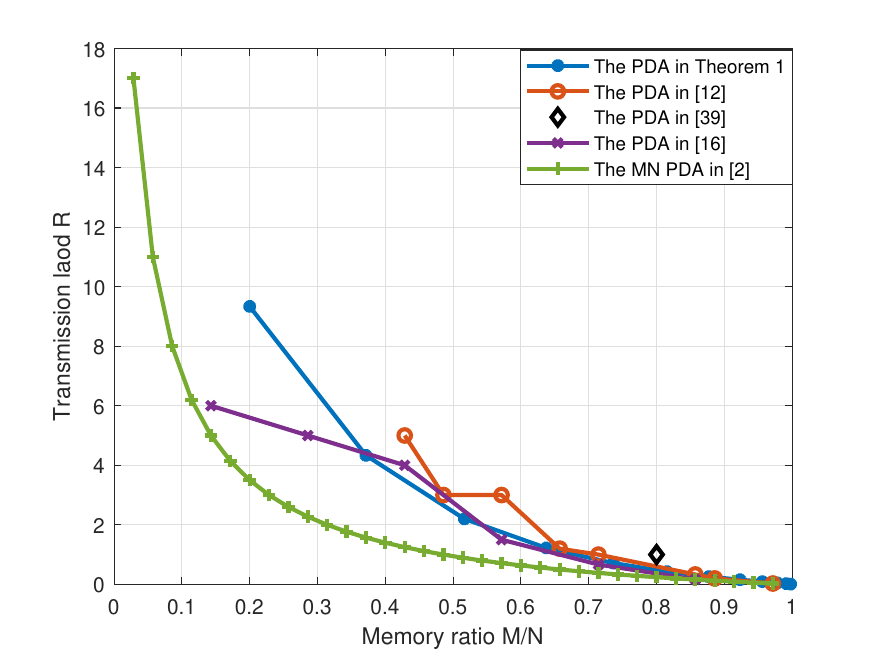}}
\caption{\small The transmission loads of PDAs in \cite{MN,CKSM,YTCC,CJYT} and Theorem \ref{th-PDA-SS}.}
\label{fig-design-R}
\end{figure}
\begin{figure}[http!]
\centerline{\includegraphics[scale=0.7]{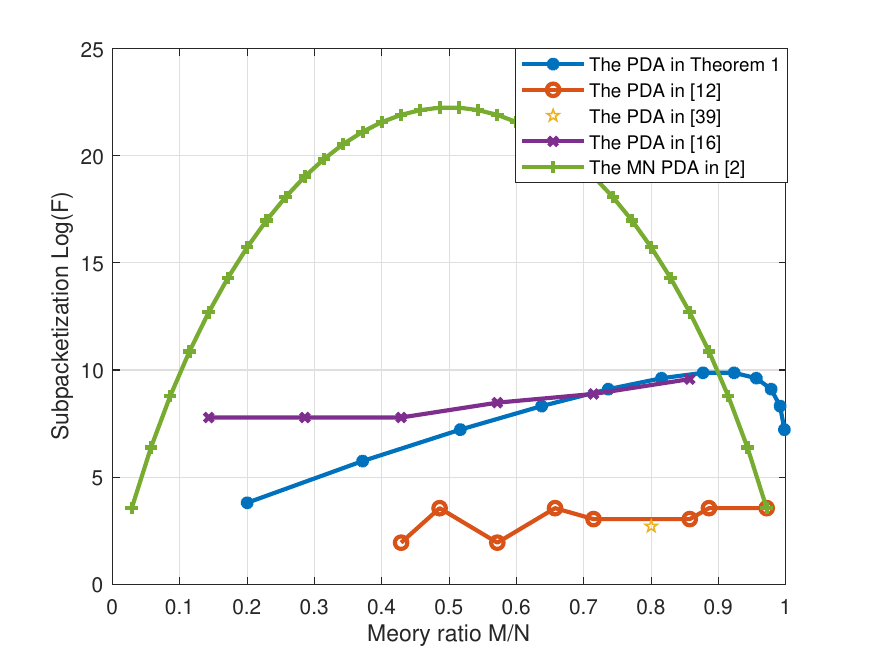}}
\caption{\small The subpacketizatins of PDAs in \cite{MN,CKSM,YTCC,CJYT} and  Corollary \ref{corollary-design}. }
\label{fig-design-F}
\end{figure}

Now let us consider the PDA in Corollary \ref{corollary-GDD}. Based on a $2$-$(m,q,3,1)$ GDD and OA$(m,q,m)$, the obtained user-delivery array is a $(K,F,Z,S)$ PDA with
\begin{align*}
K=\frac{m(m-1)q^2}{6},\ \ \ F= 3q^m,\ \ \ Z=3q^{m-3}(q^3-(q-1)^3),\ \ \  S= (q-1)^2q^{m}.
\end{align*}
Due to the special values of the user number and the memory ratio of the proposed PDA in Corollary \ref{corollary-GDD},  the comparison to the schemes in \cite{CJYT,YTCC,CKSM} is provided in Table \ref{tab-compare}.
\begin{table}
\center
\caption{The schemes in \cite{MN,CJYT,YTCC,CKSM} and Theorem \ref{th-GDD}.
\label{tab-compare}}
\renewcommand\arraystretch{0.8}
    \setlength{\tabcolsep}{0.5mm}{
\begin{tabular}{|c|c|c|c|c|}
\hline
PDAs $\&$ parameters&User number $K$& Memory ratio $M/N$
&Load $R$&Subpacketiztion $F$ \\ \hline
Corollary \ref{corollary-GDD}: $m=15,q=5$&875&0.488&5.33    &$0.92\times 10^{11}$\\
\cite{CJYT}: $m=7,q=125,z=61$   &875&0.488&64   &$3.8\times 10^{12}$\\
\cite{MN}: $K=875$, $\mu=0.488$	                    &875&0.488&1.05&$5.3\times 10^{261}$\\ \hline
Corollary \ref{corollary-GDD}: $m=16,q=4$&640&0.578&3	&$1.29\times 10^{10}$	\\
\cite{CJYT}: $m=10,q=64,z=37$   &640&0.578&13.5	&$3.6\times 10^{16}$ \\
\cite{MN}: $K=640$, $\mu=0.578$	                    &640&0.578&0.73	&$5.7\times 10^{187}$	\\ \hline
Corollary \ref{corollary-GDD}: $m=16,q=5$&1000&0.488&5.33	&$4.56\times 10^{11}$\\
\cite{CJYT}: $m=10,q=125,z=61$  &1000&0.488&64	&$4.8\times 10^{14}$\\
\cite{MN}:$K=1000$, $\mu=0.488$	                    &1000&0.488&1.05&$2.0\times 10^{299}$\\ \hline
Corollary \ref{corollary-GDD}: $m=7,q=20$&2800&0.143&120.33	&$0.38\times 10^{10}$\\
\cite{CKSM}: $q=7,m=5,t=1,a=4$   &2801&0.143&480	&$2.5\times 10^{12}$\\ \hline
Corollary \ref{corollary-GDD}: $m=4,q=8$ &128&0.33&16.33	   &$1.23\times 10^4$	\\
\cite{CKSM}: $q=3,m=5,t=1,a=4$   &121&0.33&16.2 &$7.64\times 10^{6}$	\\ \hline
Corollary \ref{corollary-GDD}: $m=1,q=3$&165&0.704&1.333  &531441	\\
\cite{YTCC}: $m=11,a=3,b=1,\eta=1$ &165&0.727&25&11	\\ \hline
Corollary \ref{corollary-GDD}: $m=4,q=12$&288&0.23&40.22& 62208	\\
\cite{YTCC}: $m=13,a=3,b=1,\eta=0$ &286&0.23&55&13	\\
\hline
\end{tabular}
}
\end{table}
In Table~\ref{tab-compare} we can see that the proposed PDAs in Corollary \ref{corollary-GDD} has the advantages both on transmission load and subpacetizations compared with the PDAs in \cite{CJYT,CKSM}; has a lower transmission load and higher subpacketization compared to the PDA  in \cite{YTCC}.
\subsection{Sketch of the proposed scheme in Theorem \ref{th-PDA-SS}}
\label{subsect-sketch-1}
Let us consider an $(L,K,\Gamma,M,N)=(3,7,7,1,7)$ coded caching scheme for the $t$-$(\Gamma,L,K,\lambda)=2$-$(7,3,7,1)$ design $(\mathcal{X},\mathfrak{B})$ (in Example \ref{exam-1-t-design}) access topology. More precisely,
user $U_{\mathcal{B}_k}$  where $k\in [7]$ is connected to the cache-nodes with indices in $\mathcal{B}_k$, where
\begin{align*}
&\mathcal{B}_1= \{1,2,4\}, \ \mathcal{B}_2= \{2,3,5\}, \ \mathcal{B}_3= \{3,4,6\}, \ \mathcal{B}_4= \{4,5,7\}, \\
&\mathcal{B}_5= \{5,6,1\},\ \mathcal{B}_6= \{6,7,2\}, \ \mathcal{B}_7= \{7,1,3\}.
\end{align*}
In addition, we have $\mu\Gamma=1$ and ${L\choose t}={3\choose 2}=3$.

Using the MN scheme and the $2$-$(7,3,7,1)$ design, we first divide each file into $7$ subfiles with equal size, i.e., for each $n\in [7]$, $W_n=(W_{n,\mathcal{D}})_{\mathcal{D}\in {[7]\choose 1}}$, and further divide each subfile into $3$ packets with equal size, i.e., for each $n\in [N]$ and $\mathcal{D}\in {[7]\choose 1}$, $W_{n,\mathcal{D}}=(W^{\mathcal{T}}_{n,\mathcal{D}})_{\mathcal{T}\in {[3]\choose 2}}$.
Let us introduce our main idea by constructing three arrays $\mathbf{C}$, $\mathbf{U}$, and $\mathbf{Q}$ defined in Definition \ref{defn:three arrays}.

By Definition \ref{defn:three arrays}, we have $F={\Gamma\choose \mu\Gamma}{L\choose t}={7\choose 1}{3\choose 2}=21$ and $K={\Gamma\choose t}/{L\choose t}={7\choose 2}/{3\choose 2}=7$. For convenience, we will represent the row labels of $\mathbf{C}$, $\mathbf{U}$, and $\mathbf{Q}$ by a pair $(\mathcal{D},\mathcal{T})$ where $\mathcal{D}\in {[\Gamma]\choose \mu\Gamma}={[7]\choose 1}$, $\mathcal{T}\in {[L]\choose 2}={[3]\choose 2}$.

\subsubsection{Construction of node-placement array $\mathbf{C}$}
We divide each file $W_n$ where $n\in [7]$ into $21$ packets, where
 each packet is denoted by $W^{\mathcal{T}}_{n,\mathcal{D}}$ where $\mathcal{T}\in {[3]\choose 2}$ and $\mathcal{D}\in {[7]\choose 1}$,
and  is cached by cache-node $\gamma$    if $\gamma\in \mathcal{D}$. Then the packets cached by cache-nodes can be written as (note that in the examples for simplicity we omit the braces and commas when we indicate the indices of cache-nodes, users and packets)
\begin{align}
\mathcal{Z}_{C_1}&=\{W^{12}_{n,1},\ W^{13}_{n,1},\ W^{23}_{n,1}\ |\ n\in [7]\},\ \
\mathcal{Z}_{C_2} =\{W^{12}_{n,2},\ W^{13}_{n,2},\ W^{23}_{n,2}\ |\ n\in [7]\}, \nonumber\\
\mathcal{Z}_{C_3}&=\{W^{12}_{n,3},\ W^{13}_{n,3},\ W^{23}_{n,3}\ |\  n\in [7]\},\ \  \mathcal{Z}_{C_4} =\{W^{12}_{n,4},\ W^{13}_{n,4},\ W^{23}_{n,4}\ |\ n\in [7]\},\nonumber\\\
\mathcal{Z}_{C_5}&=\{W^{12}_{n,5},\ W^{13}_{n,5},\ W^{23}_{n,5}\ |\ n\in [7]\},\ \
\mathcal{Z}_{C_6} =\{W^{12}_{n,6},\ W^{13}_{n,6},\ W^{23}_{n,6}\ |\ n\in [7]\},\nonumber\\
\mathcal{Z}_{C_7}&=\{W^{12}_{n,7},\ W^{13}_{n,7},\ W^{23}_{n,7\}}\ |\ n\in [7]\}.\label{eq-exa-design-cache-node}
\end{align}
By Definition \ref{defn:three arrays}, a $21\times 7$ node-placement array $\mathbf{C}$ which represents the packets cached by the cache-nodes can be written in Table \ref{tab-caching-node} as follows.
Note that since each row of $\mathbf{C}$  represents a packet $W^{\mathcal{T}}_{n,\mathcal{D}}$ for the files $n\in [N]$, we let the row index of $\mathbf{C}$ be $(\mathcal{D},\mathcal{T})$, where $\mathcal{T}\in {[3]\choose 2}$ and $\mathcal{D}\in {[7]\choose 1}$.

\begin{table}
\center
\caption{Node-placement array $\mathbf{C}$.
\label{tab-caching-node}}
\renewcommand\arraystretch{0.8}
\begin{tabular}{|c|ccccccc|}
\hline
Packet labels&\multicolumn{7}{|c|}{Cache-node set $\mathcal{X}$} \\ \hline
$\mathcal{D},\mathcal{T}$         &  $C_1$&$C_2$	&$C_3$	&$C_4$	&$C_5$	&$C_6$	&$C_7$ \\ \hline
$\{1\},\{1,2\}$&$*$&&&&&&\\
$\{2\},\{1,2\}$&&$*$&&&&&\\
$\{3\},\{1,2\}$&&&$*$&&&&\\
$\{4\},\{1,2\}$&&&&$*$&&&\\
$\{5\},\{1,2\}$&&&&&$*$&&\\
$\{6\},\{1,2\}$&&&&&&$*$&\\
$\{7\},\{1,2\}$&&&&&&&$*$\\ \hline
$\{1\},\{1,3\}$&$*$&&&&&&\\
$\{2\},\{1,3\}$&&$*$&&&&&\\
$\{3\},\{1,3\}$&&&$*$&&&&\\
$\{4\},\{1,3\}$&&&&$*$&&&\\
$\{5\},\{1,3\}$&&&&&$*$&&\\
$\{6\},\{1,3\}$&&&&&&$*$&\\
$\{7\},\{1,3\}$&&&&&&&$*$\\ \hline
$\{1\},\{2,3\}$&$*$&&&&&&\\
$\{2\},\{2,3\}$&&$*$&&&&&\\
$\{3\},\{2,3\}$&&&$*$&&&&\\
$\{4\},\{2,3\}$&&&&$*$&&&\\
$\{5\},\{2,3\}$&&&&&$*$&&\\
$\{6\},\{2,3\}$&&&&&&$*$&\\
$\{7\},\{2,3\}$&&&&&&&$*$\\ \hline
\end{tabular}
\end{table}
\subsubsection{Construction of user-retrieve array $\mathbf{U}$}
Recall that the $2$-$(7,3,7,1)$ design $(\mathcal{X},\mathfrak{B})$ in Example \ref{exam-1-t-design} has the blocks $\mathfrak{B}=\{\{1,2,4\}$, $\{2,3,5\}$, $\{3,4,6\}$, $\{4,5,6\}$, $\{5,6,1\}$, $\{6,7,2\}$, $\{7,1,3\}\}$. Then the users $U_{\mathcal{B}}$ where $\mathcal{B}\in\mathfrak{B}$ can retrieve the following  packets,
\begin{align}
&\mathcal{Z}_{U_{124}}=\{{\color{red}W^{12}_{n,1},W^{13}_{n,1},W^{23}_{n,1}}, {\color{blue}W^{12}_{n,2},W^{13}_{n,2},W^{23}_{n,2}},
{\color{magenta}W^{12}_{n,4},W^{13}_{n,4},W^{23}_{n,4}}|n\in [7]\},\nonumber\\
&\mathcal{Z}_{U_{235}}=\{{\color{red}W^{12}_{n,2},W^{13}_{n,2}},{\color{red}W^{23}_{n,2}}, {\color{blue}W^{12}_{n,3},W^{13}_{n,3},W^{23}_{n,3}},
{\color{magenta}W^{12}_{n,5},W^{13}_{n,5},W^{23}_{n,5}}|n \in[7]\},\nonumber\\
& \mathcal{Z}_{U_{346}}= \{{\color{red}W^{12}_{n,3},W^{13}_{n,3},W^{23}_{n,3}}, {\color{blue}W^{12}_{n,4},W^{13}_{n,4}},
{\color{blue} W^{23}_{n,4}},{\color{magenta}W^{12}_{n,6},W^{13}_{n,6},W^{23}_{n,6}} | n\in [7]\}\nonumber\\
&\mathcal{Z}_{U_{457}}=\{{\color{red}W^{12}_{n,4}},  {\color{red}W^{13}_{n,4},W^{23}_{n,4}},{\color{blue}W^{12}_{n,5},W^{13}_{n,5},W^{23}_{n,5}},
{\color{magenta}W^{12}_{n,7},W^{13}_{n,7},W^{23}_{n,7}}  | n\in [7]\},\nonumber\\
&\mathcal{Z}_{U_{156}}=\{ {\color{red}W^{12}_{n,5},W^{13}_{n,5},W^{23}_{n,5}}, {\color{blue}W^{12}_{n,6}, W^{13}_{n,6}},
{\color{blue}W^{23}_{n,6}}, {\color{magenta}W^{12}_{n,1},W^{13}_{n,1},W^{23}_{n,1}}| n\in [7]\},\nonumber\\
&\mathcal{Z}_{U_{267}}=\{{\color{red}W^{12}_{n,6}},{\color{red}W^{13}_{n,6},W^{23}_{n,6}}, {\color{blue}W^{12}_{n,7},W^{13}_{n,7},W^{23}_{n,7}},
{\color{magenta}W^{12}_{n,2},W^{13}_{n,2}},{\color{magenta}W^{23}_{n,2}}|n\in [7]\}, \nonumber\\
&\mathcal{Z}_{U_{137}}=\{{\color{red}W^{12}_{n,7},W^{13}_{n,7},W^{23}_{n,7}}, {\color{blue}W^{12}_{n,1}},{\color{blue}W^{13}_{n,1},W^{23}_{n,1}},
{\color{magenta}W^{12}_{n,3},W^{13}_{n,3},W^{23}_{n,3}}| n\in [7]\}.\label{eq-exa-design-cache-user}
\end{align}
By Definition \ref{defn:three arrays}, a $21\times 7$ node-placement array $\mathbf{C}$ representing the packets retrieved by the users can be written in Table \ref{tab-user-retrieve} as follows.
\begin{table}[!htbp]
\center
\caption{User-retrieve array $\mathbf{U}$.
\label{tab-user-retrieve}}
\renewcommand\arraystretch{0.8}
\begin{tabular}{|c|ccccccc|}
\hline
Subpacket labels&\multicolumn{7}{|c|}{User set $\mathfrak{B}$} \\ \hline
$(\mathcal{D},\mathcal{T})$         &  $U_{124}$&$U_{235}$&$U_{346}$&$U_{457}$&$U_{156}$&$U_{267}$&$U_{137}$\\ \hline
$(\{1\},\{1,2\})$&*&&&&*&&*\\
$(\{2\},\{1,2\})$&*&*&&&&*&\\
$(\{3\},\{1,2\})$&&*&*&&&&*\\
$(\{4\},\{1,2\})$&*&&*&*&&&\\
$(\{5\},\{1,2\})$&&*&&*&*&&\\
$(\{6\},\{1,2\})$&&&*&&*&*&\\
$(\{7\},\{1,2\})$&&&&*&&*&*\\ \hline
$(\{1\},\{1,3\})$&*&&&&*&&*\\
$(\{2\},\{1,3\})$&*&*&&&&*&\\
$(\{3\},\{1,3\})$&&*&*&&&&*\\
$(\{4\},\{1,3\})$&*&&*&*&&&\\
$(\{5\},\{1,3\})$&&*&&*&*&&\\
$(\{6\},\{1,3\})$&&&*&&*&*&\\
$(\{7\},\{1,3\})$&&&&*&&*&*\\ \hline
$(\{1\},\{2,3\})$&*&&&&*&&*\\
$(\{2\},\{2,3\})$&*&*&&&&*&\\
$(\{3\},\{2,3\})$&&*&*&&&&*\\
$(\{4\},\{2,3\})$&*&&*&*&&&\\
$(\{5\},\{2,3\})$&&*&&*&*&&\\
$(\{6\},\{2,3\})$&&&*&&*&*&\\
$(\{7\},\{2,3\})$&&&&*&&*&*\\ \hline
\end{tabular}
\end{table}

\subsubsection{Construction of user-delivery array $\mathbf{Q}$}
For the sake of clarity, in the user-delivery array each non-null entry is filled by a set, instead of an integer.
As shown in Table \ref{tab-user-retrieve}, each entry $\mathbf{U}((\mathcal{D},\mathcal{T}),\mathcal{B})$ is null if and only if $\mathcal{D}\bigcap \mathcal{B}=\emptyset$ where $\mathcal{D}\in {[7]\choose 1}$ and $\mathcal{B}\in\mathfrak{B}$. Then we put the subset $\mathcal{D}\cup\mathcal{B}(\mathcal{T})$ into the entry $\mathbf{U}((\mathcal{D},\mathcal{T}),\mathcal{B})$ to the user-delivery array $\mathbf{Q}$.  For example, we can see that the entries $\mathbf{U}((\{3\},\{1,2\}),\{1,2,4\})$,  $\mathbf{U}((\{1\},\{1,2\}),\{2,3,5\})$ and $\mathbf{U}((\{2\},\{1,2\}),\{1,3,7\})$ are null in Table \ref{tab-user-retrieve}. Then we put the subset
\begin{align*}
\{1,2,3\} =&\{3\}\cup\{1,2,4\}(\{1,2\})=\{3\}\cup\{1,2\}\\
=&\{1\}\cup\{2,3,5\}(\{1,2\})=\{1\}\cup\{2,3\}\\
=&\{2\}\cup\{1,3,7\}(\{1,2\})=\{2\}\cup\{1,3\}
\end{align*} into the entries $\mathbf{U}((\{3\},\{1,2\}),\{1,2,4\})$, $\mathbf{U}((\{1\},\{1,2\}),\{2,3,5\})$ and $\mathbf{U}((\{2\},\{1,2\}),\{1,3,7\})$. We can check that the entries
\begin{align*}
&\mathbf{U}((\{3\},\{1,2\}),\{2,3,5\})=
\mathbf{U}((\{3\},\{1,2\}),\{1,3,7\})\\
=&\mathbf{U}((\{1\},\{1,2\}),\{1,2,4\})=
\mathbf{U}((\{1\},\{1,2\}),\{1,3,7\})\\
=&\mathbf{U}((\{2\},\{1,2\}),\{1,2,4\})=
\mathbf{U}((\{2\},\{1,2\}),\{2,3,5\})=*.
\end{align*} That is, the condition C3 of Definition \ref{def-PDA} holds.
  Similarly in this example we can obtain the $21\times 7$ user-delivery array $\mathbf{Q}$ listed in Table \ref{tab-user-delivery}.
 We can check that there are exactly $S={\Gamma \choose t+\mu\Gamma}-K{L\choose t+\mu\Gamma}={7\choose 2+1}-7\times {3\choose 2+1}=35-7=28$, and the obtained $\mathbf{Q}$ is a $(7,21,9,28)$ PDA which leads to a $(7,M',N)$ coded caching scheme for a shared-link network with $M'/N=3/7$, subpacketization $21$ and transmission load $28/21=4/3$ by Lemma \ref{le-Fundamental}. The total coded caching gain is $K(F-Z)/S=7\times 4\times 3/28=3$.
\begin{table}[!htbp]
\center
\caption{User-delivery array $\mathbf{Q}$.
\label{tab-user-delivery}}
\renewcommand\arraystretch{0.2}
    \setlength{\tabcolsep}{0.5mm}{
\begin{tabular}{|c|ccccccc|}
\hline
Subpacket labels&\multicolumn{7}{|c|}{User set $\mathfrak{B}$} \\ \hline
$\mathcal{D},\mathcal{T}$         &
$U_{124}$&$U_{235}$&$U_{346}$&$U_{457}$&
$U_{156}$&$U_{267}$&$U_{137}$\\ \hline
$\{1\},\{1,2\}$&*&$123$&$134$&$145$&*&$126$&*\\
$\{2\},\{1,2\}$&*&*&$234$&$245$&$125$&*&$123$\\
$\{3\},\{1,2\}$&$123$&*&*&$345$&$135$&$236$&*\\
$\{4\},\{1,2\}$&*&$234$&*&*&$145$&$246$&$134$\\
$\{5\},\{1,2\}$&$125$&*&$345$&*&*&$256$&$135$\\
$\{6\},\{1,2\}$&$126$&$236$&*&$456$&*&*&$136$\\
$\{7\},\{1,2\}$&$127$&$237$&$347$&*&$157$&*&*\\ \hline
$\{1\},\{1,3\}$&*&$125$&$136$&$147$&*&$127$&*\\
$\{2\},\{1,3\}$&*&*&$236$&$247$&$126$&*&$127$\\
$\{3\},\{1,3\}$&$134$&*&*&$347$&$136$&$237$&*\\
$\{4\},\{1,3\}$&*&$245$&*&*&$146$&$247$&$147$\\
$\{5\},\{1,3\}$&$145$&*&$356$&*&*&$257$&$157$\\
$\{6\},\{1,3\}$&$146$&$256$&*&$467$&*&*&$167$\\
$\{7\},\{1,3\}$&$147$&$257$&$367$&*&$167$&*&*\\ \hline
$\{1\},\{2,3\}$&*&$135$&$146$&$157$&*&$167$&*\\
$\{2\},\{2,3\}$&*&*&$246$&$257$&$256$&*&$237$\\
$\{3\},\{2,3\}$&$234$&*&*&$357$&$356$&$367$&*\\
$\{4\},\{2,3\}$&*&$345$&*&*&$456$&$467$&$347$\\
$\{5\},\{2,3\}$&$245$&*&$456$&*&*&$567$&$357$\\
$\{6\},\{2,3\}$&$246$&$356$&*&$567$&*&*&$367$\\
$\{7\},\{2,3\}$&$247$&$357$&$467$&*&$567$&*&*\\ \hline
\end{tabular}
}
\end{table}

\subsection{Sketch of the proposed scheme in Theorem \ref{th-GDD}}
\label{subsect-sketch-2}
Now we let the access topology be the $t$-$(m,q,L,\lambda)=2$-$(3,2,2,1)$ GDD
$(\mathcal{X},\mathfrak{G},\mathfrak{B})$ where
\begin{subequations}
\label{eq-2-3-2-2-1-GDD}
\begin{align}
\mathcal{X}&=\{(1,1),(1,2),(2,1),(2,2),(3,1),(3,2)\},\label{eq-point-GDD}\\
\mathfrak{G}&=\{\mathcal{G}_1=\{(1,1),(1,2)\},\ \ \mathcal{G}_2=\{(2,1),(2,2)\},\ \
\mathcal{G}_3=\{(3,1),(3,2)\}\},\label{eq-point-GDD2}\\
\mathfrak{B}&=\{\ \mathcal{B}_1=\{(1,1),(2,1)\},\ \ \mathcal{B}_2=\{(1,1),(2,2)\},\ \ \mathcal{B}_3=\{(1,1),(3,1)\},\nonumber\\
\ \ &\ \ \ \ \ \mathcal{B}_4=\{(1,1),(3,2)\},
\ \ \mathcal{B}_5=\{(1,2),(2,1)\},\ \
\mathcal{B}_6=\{(1,2),(2,2)\},\nonumber\\
\ \ &\ \ \ \ \ \mathcal{B}_7=\{(1,2),(3,1)\},\ \ \mathcal{B}_8=\{(1,2),(3,2)\},
\ \ \mathcal{B}_9=\{(2,1),(3,1)\},\nonumber\\
\ \ &\ \ \ \ \
\mathcal{B}_{10}=\{(2,1),(3,2)\},\  \mathcal{B}_{11}=\{(2,2),(3,1)\},\   \mathcal{B}_{12}=\{(2,2),(3,2)\}\
\}.\label{eq-block-GDD}
\end{align}
\end{subequations}
The topology represented by  the above  $t$-GDD, contains $6$ cache-nodes and $12$ users.
Each cache-node is represented by a vector in $\mathcal{X}$.
Each user $U_{\mathcal{B}_k}$ where $k\in [12]$ is connected to the cache-nodes with indices in $\mathcal{B}_k$  as shown in~\eqref{eq-block-GDD}.
In the following, for the $2$-$(3,2,2,1)$ GDD $(\mathcal{X},\mathfrak{G},\mathfrak{B})$ access topology, we will design a $(L,K,\Gamma,M,N)=(2,12,6,3,6)$ coded caching scheme  whose cache placement is  based on the OA$(m=3,q=2,s=2)$,  denoted by $\mathbf{A}=(\mathbf{A}(j,u))_{j\in [4],u\in [3]}$ (shown in~\eqref{eq-OA-4} of Example~\ref{exam-OA-3-2-3}).

 We will also introduce our main idea by constructing three arrays $\mathbf{C}$, $\mathbf{U}$, and $\mathbf{Q}$ as follows.

\subsubsection{Construction of node-placement array $\mathbf{C}$}
We divide each file into $\binom{L}{t}q^s=4$   subfiles with equal size, i.e., for each $n\in [6]$, $W_n=\left(W_{n,j}: \mathcal{T}\in {[L]\choose t},j\in [q^s] \right)=\left(W_{n,j}: \mathcal{T}\in {[2]\choose 2},j\in [4] \right)$. Each cache-node $C_{u,v}$ where $u\in [3]$ and $v\in [2]$ caches the packet $W^{[2]}_{n,j}$ if   $\mathbf{A}(j,u)=v$ for each $j\in [4]$. So the packets cached by the cache-nodes can be written as follows.
\begin{align}
\mathcal{Z}_{C_{1,1}}=\left\{W^{[2]}_{n,1},W^{[2]}_{n,3}\ |\  n\in [6]\right\},\ \ \
\mathcal{Z}_{C_{1,2}}=\left\{W^{[2]}_{n,2},W^{[2]}_{n,4}\ |\  n\in [6]\right\},\nonumber\\
\mathcal{Z}_{C_{2,1}}=\left\{W^{[2]}_{n,1},W^{[2]}_{n,2}\ |\  n\in [6]\right\},\ \ \
\mathcal{Z}_{C_{2,2}}=\left\{W^{[2]}_{n,3},W^{[2]}_{n,4}\ |\  n\in [6]\right\},\nonumber\\
\mathcal{Z}_{C_{3,1}}=\left\{W^{[2]}_{n,1},W^{[2]}_{n,4}\ |\  n\in [6]\right\},\ \ \
\mathcal{Z}_{C_{3,2}}=\left\{W^{[2]}_{n,2},W^{[2]}_{n,3}\ |\  n\in [6]\right\}. \label{eq-cache-subpackets-exam}
\end{align}
We can see that each cache-node caches $2\times 6=12$  packets. By Definition \ref{defn:three arrays}, the $4\times 6$ node-placement array $\mathbf{C}$ representing the packets cached by the cache-nodes can be written in Table~\ref{tab-caching-node-OA}. Note that the row index of $\mathbf{C}$ is $(j,[2])$ where $j\in [4]$.

\begin{table}
\center
\caption{Node-placement array $\mathbf{C}$.
\label{tab-caching-node-OA}}
\begin{tabular}{|c|c|cccccc|} \hline
Array&Packet labels&\multicolumn{6}{|c|}{cache-nodes} \\ \hline
$\mathbf{A}$&$(j\in [4],\mathcal{T}=[2])$
           &$C_{1,1}$&$C_{1,2}$&$C_{2,1}$&$C_{2,2}$&$C_{3,1}$&$C_{3,2}$ \\ \hline
111 &$(1,[2])$&   *& &*& &*&\\
212 &$(2,[2])$& &*&*& &&*\\
122 &$(3,[2])$&*& & &*&&*\\
221 &$(4,[2])$ & &*& &*&*&\\ \hline
\end{tabular}
\end{table}

\subsubsection{Construction of user-retrieve array $\mathbf{U}$}
Using the $2$-$(3,2,3,1)$ GDD $(\mathcal{X},\mathfrak{G},\mathfrak{B})$ in Example \ref{exam-dual-8-4}, the users can retrieve the following packets.
\begin{eqnarray}\label{eq-retrieve-subpackets-exam}
\begin{split}
\mathcal{Z}_{U_{(1,1),(2,1)}}&=\mathcal{Z}_{C_{1,1}}\bigcup \mathcal{Z}_{C_{2,1}}=\{
W^{[2]}_{n,1},W^{[2]}_{n,2},W^{[2]}_{n,3}\ \Big|\  n\in [6]\},&\\
\mathcal{Z}_{U_{(1,2),(2,1)}}&=\mathcal{Z}_{C_{1,2}}\bigcup\mathcal{Z}_{C_{2,1}}=\{
W^{[2]}_{n,1},W^{[2]}_{n,2},W^{[2]}_{n,4}\ \Big|\  n\in [6]\},&\\
\mathcal{Z}_{U_{(1,1),(2,2)}}&=\mathcal{Z}_{C_{1,1}}\bigcup\mathcal{Z}_{C_{2,2}}=\{
W^{[2]}_{n,1},W^{[2]}_{n,3},W^{[2]}_{n,4} \ \Big|\  n\in [6]\},&\\
\mathcal{Z}_{U_{(1,2),(2,2)}}&=\mathcal{Z}_{C_{1,2}}\bigcup\mathcal{Z}_{C_{2,2}}=\{
W^{[2]}_{n,2},W^{[2]}_{n,3},W^{[2]}_{n,4} \ \Big|\  n\in [6]\},&\\
\mathcal{Z}_{U_{(1,1),(3,1)}}&=\mathcal{Z}_{C_{1,1}}\bigcup\mathcal{Z}_{C_{3,1}}=\{
W^{[2]}_{n,1},W^{[2]}_{n,3},W^{[2]}_{n,4} \ \Big|\  n\in [6]\},&\\
\mathcal{Z}_{U_{(1,2),(3,1)}}&=\mathcal{Z}_{C_{1,2}}\bigcup\mathcal{Z}_{C_{3,1}}=\{
W^{[2]}_{n,1},W^{[2]}_{n,2},W^{[2]}_{n,4} \ \Big|\  n\in [6]\},&\\
\mathcal{Z}_{U_{(1,1),(3,2)}}&=\mathcal{Z}_{C_{1,1}}\bigcup\mathcal{Z}_{C_{3,2}}=\{
W^{[2]}_{n,1},W^{[2]}_{n,2},W^{[2]}_{n,3} \ \Big|\  n\in [6]\},&\\
\mathcal{Z}_{U_{(1,2),(3,2)}}&=\mathcal{Z}_{C_{1,2}}\bigcup\mathcal{Z}_{C_{3,2}}=\{
W^{[2]}_{n,2},W^{[2]}_{n,3},W^{[2]}_{n,4} \ \Big|\  n\in [6]\},&\\
\mathcal{Z}_{U_{(2,1),(3,1)}}&=\mathcal{Z}_{C_{2,1}}\bigcup\mathcal{Z}_{C_{3,1}}=\{
W^{[2]}_{n,1},W^{[2]}_{n,2},W^{[2]}_{n,4} \ \Big|\  n\in [6]\},&\\
\mathcal{Z}_{U_{(2,2),(3,1)}}&=\mathcal{Z}_{C_{2,2}}\bigcup\mathcal{Z}_{C_{3,1}}=\{
W^{[2]}_{n,1},W^{[2]}_{n,3},W^{[2]}_{n,4} \ \Big|\  n\in [6]\},&\\
\mathcal{Z}_{U_{(2,1),(3,2)}}&=\mathcal{Z}_{C_{2,1}}\bigcup\mathcal{Z}_{C_{3,2}}=\{
W^{[2]}_{n,1},W^{[2]}_{n,2},W^{[2]}_{n,3} \ \Big|\  n\in [6]\},&\\
\mathcal{Z}_{U_{(2,2),(3,2)}}&=\mathcal{Z}_{C_{2,2}}\bigcup\mathcal{Z}_{C_{3,2}}=\{
W^{[2]}_{n,2},W^{[2]}_{n,3},W^{[2]}_{n,4} \ \Big|\  n\in [6]\}.&
\end{split}
\end{eqnarray}
By Definition \ref{defn:three arrays} we have a $4\times 12$ user-retrieve array $\mathbf{U}$ to represent the packets retrieved by the users in Table \ref{tab-user-retrieve-GDD} as follows.

\begin{table}
\center
\caption{User-retrieve array $\mathbf{U}$. We omit the    braces and commas in the vectors; e.g., $111$ represents $(1,1,1)$, $U_{11,21}$ represents $U_{(1,1),(2,1)}$.
\label{tab-user-retrieve-GDD}}
\renewcommand\arraystretch{0.8}
    \setlength{\tabcolsep}{1.36mm}{
\begin{tabular}{|c|c|cccccccccccc|}
Array&Packet labels&\multicolumn{12}{|c|}{User set $U_{\mathcal{B}}$, $\mathcal{B}\in\mathfrak{B}$} \\ \hline
$\mathbf{A}$&$j\in [4],\mathcal{T}=[2]$       &$U_{11,21}$&$U_{11,22}$&$U_{11,31}$&$U_{11,32}$&$U_{12,21}$&$U_{12,22}$&$U_{12,31}$&$U_{12,32}$
&$U_{21,31}$&$U_{21,32}$&$U_{22,31}$&$U_{22,32}$\\\hline
111&$1,[2]$&*&*&*&*&*& &*& &*&*&*& \\\hline
212&$2,[2]$&*& & &*&*&*&*&*&*&*& &*\\\hline
122&$3,[2]$&*&*&*&*& &*& &*& &*&*&*\\\hline
221&$4,[2]$& &*&*& &*&*&*&*&*& &*&*\\\hline
\end{tabular}
}
\end{table}
\subsubsection{Construction of user-delivery array $\mathbf{Q}$}
Using the  OA$(m=3,q=2,s=2)$ $\mathbf{A}=(\mathbf{A}(j,u))_{j\in [4],u\in [3]}$ in~\eqref{eq-OA-4} of Example~\ref{exam-OA-3-2-3} and the GDD $(\mathcal{X},\mathfrak{G},\mathfrak{B})$ in \eqref{eq-2-3-2-2-1-GDD}, we can put the null entries of $\mathbf{U}$ in the following way to obtain the use-delivery array $\mathbf{Q}$, i.e., Table \ref{tab-user-delivery-GDD}.  For the sake of clarity, each block $\mathcal{B}_k, k\in [12]$  in \eqref{eq-2-3-2-2-1-GDD}  can be represented by $\{(u_{k,1},v_{k,1}), (u_{k,2},v_{k,2})\}$ where $u_{k,1}$, $u_{k,2}\in [3]$ and $v_{k,1}$, $v_{k,2}\in [2]$, and each non-null entry of $\mathbf{U}$ is filled by a vector, instead of an integer. As shown in Table \ref{tab-user-retrieve-GDD}, for any integer $j\in [4]$ and integer $k\in [12]$, the entry $\mathbf{U}((j,[2]),\mathcal{B}_k\})$ is null if and only if $\text{d}(\mathbf{A}(j,\{u_{k,1},u_{k,2}\}),(v_{k,1},v_{k,2}))=2$.
Then we put the vector ${\bf e}=(e_1,e_2,e_3)$ where
\begin{align*}
{\bf e}(\{u_{k,1},u_{k,2}\})=(v_{k,1},v_{k,2})\ \ \text{and}\ \
{\bf e}([3]\setminus\{u_{k,1},u_{k,2}\})
=\mathbf{A}(j,[3]\setminus\{u_{k,1},u_{k,2}\})
\end{align*} into the entry $\mathbf{U}((j,[2]),\mathcal{B}_k)$. For example, let us see the columns of the user-retrieve array $\mathbf{U}$ in Table \ref{tab-user-retrieve-GDD} which are labeled by the blocks
\begin{align*}
\mathcal{B}_2&=\{(u_{2,1},v_{2,1})=(1,1), (u_{2,2},v_{2,2})=(2,2)\}\\
\mathcal{B}_4&=\{(u_{4,1},v_{4,1})=(1,1), (u_{4,2},v_{4,2})=(3,2)\}\\
\mathcal{B}_{12}&=\{(u_{12,1},v_{12,1})=(2,2), (u_{12,2},v_{12,2})=(3,2)\}
\end{align*}respectively. We can see that the entries
\begin{align*}
\mathbf{U}((2,[2]),\mathcal{B}_2),\ \
\mathbf{U}((4,[2]),\mathcal{B}_4),\ \
\mathbf{U}((1,[2]),\mathcal{B}_{12})
\end{align*} are null in Table \ref{tab-user-retrieve-GDD}. Then we put the vector $(1,2,2)$ into these three entries since
\begin{align*}
&\left\{\begin{array}{c}
(1,2,2)(\{u_{2,1},u_{2,2}\})=(1,2,2)(\{1,2\})=(1,2)=(v_{2,1},v_{2,2}),\\
(1,2,2)([3]\{1,2\})=(1,2,2)(3)=2=\mathbf{A}(2,3);\ \ \ \ \ \ \ \ \ \ \ \ \ \ \ \ \
\end{array} \right.\\
&\left\{\begin{array}{c}
(1,2,2)(\{u_{4,1},u_{4,2}\})=(1,2,2)(\{1,3\})=(1,2)=(v_{4,1},v_{4,2}),\\
(1,2,2)([3]\{1,3\})=(1,2,2)(2)=2=\mathbf{A}(4,2);\ \ \ \ \ \ \ \ \ \ \ \ \ \ \ \ \
\end{array}\right.\\
&\left\{\begin{array}{c}
(1,2,2)(\{u_{12,1},u_{12,2}\})=(1,2,2)(\{2,3\})=(2,2)=(v_{2,1},v_{2,2}),\\
(1,2,2)([3]\{2,3\})=(1,2,2)(1)=1=\mathbf{A}(1,1).\ \ \ \ \ \ \ \ \ \ \ \ \ \ \ \ \
\end{array}\right.
\end{align*}One can check that then entries
\begin{align*}
&\mathbf{U}((2,[2]),\mathcal{B}_8)=
\mathbf{U}((2,[2]),\mathcal{B}_{12})
=\mathbf{U}((4,[2]),\mathcal{B}_2)=
\mathbf{U}((4,[2]),\mathcal{B}_{12})\\
=&\mathbf{U}((1,[2]),\mathcal{B}_{2})=
\mathbf{U}((1,[2]),\mathcal{B}_{4})=*.
\end{align*} That is, the condition C3 of Definition \ref{def-PDA} holds. We can check that $\mathbf{Q}$ is a $(12, 4,3,4)$ PDA.
\begin{table}
\center
\caption{User-delivery array $\mathbf{U}$. We omit the    braces and commas in the vectors; e.g., $111$ represents $(1,1,1)$, $U_{11,21}$ represents $U_{(1,1),(2,1)}$.
\label{tab-user-delivery-GDD}}
\renewcommand\arraystretch{1}
    \setlength{\tabcolsep}{1.36mm}{
\begin{tabular}{|c|c|cccccccccccc|}
Array&Packet labels&\multicolumn{12}{|c|}{User set $U_{\mathcal{B}}$, $\mathcal{B}\in\mathfrak{B}$} \\ \hline
$\mathbf{A}$&$(j\in [4],\mathcal{T}=[2])$
&$U_{11,21}$&$U_{11,22}$&$U_{11,31}$&$U_{11,32}$&$U_{12,21}$&$U_{12,22}$&$U_{12,31}$&$U_{12,32}$
&$U_{21,31}$&$U_{21,32}$&$U_{22,31}$&$U_{22,32}$\\\hline
111&$(1,[2])$&*&*&*&*&*& 221&*& 212&*&*&*&122 \\\hline
212&$(2,[2])$&*&122 &111 &*&*&*&*&*&*&*&221 &*\\\hline
122&$(3,[2])$&*&*&*&*&212 &*& 221&*& 111&*&*&*\\\hline
221&$(4,[2])$&111 &*&*&122 &*&*&*&*&*&212 &*&*\\\hline
\end{tabular}
}
\end{table}

\section{Multiaccess coded caching scheme for the $t$-design access topology: Proof of Theorem~\ref{th-PDA-SS}}
\label{sec-packing-design}
In this section, we describe   the $(L,K,\Gamma,M,N)$ coded caching scheme for the   $t$-$(\Gamma,L,K,1)$ design $(\mathcal{X},\mathfrak{B})$ access topology for Theorem~\ref{th-PDA-SS}.  By Definition of $t$-design, we have $
K={\Gamma\choose t}/{L\choose t}$. We will propose a coded caching scheme with a placement phase that depends only on $(\Gamma,M,N)$ MN scheme and is
agnostic to the access topology.\footnote{\label{foot:practice} This is also practical in real systems because placement is always done at peak times when we only know the number of cache-nodes (which are fixed edges in the network) and do not know the network access topology exactly. }
Our delivery phase is designed with the knowledge of the access topology. In addition, it is a unified delivery phase for all topologies of $t$-design, given a $t$.
It can also be seen that the load decreases as $t$ increases. This corresponds to the fact that,    $t_1$-design contains $t_2$-design if $t_1<t_2$, and that the topologies in $t_2$-design have denser  connectivity than the topologies in $t_1$-design which are not  $t_2$-design.

In the following, we describe our scheme by constructing sequentially  node-placement,  user-retrieve, and  user-delivery arrays.
\subsection{Node-placement array}
 The server places the files on $\Gamma$ cache-nodes using the placement strategy of the MN scheme. That is, the server divides each file into ${\Gamma\choose \mu\Gamma}$ subfiles of equal size, i.e., for each $n\in [N]$, $W_n=(W_{n,\mathcal{D}})_{\mathcal{D}\in {[\Gamma]\choose \mu\Gamma}}$. Each cache-node $C_{ \gamma}$ where $\gamma\in [\Gamma]$ caches the subfiles
    \begin{align}\label{eq-cache-packets}
    \mathcal{Z}_{C_\gamma}=\left\{W_{n,\mathcal{D}}\ \Big|\ \gamma\in\mathcal{D}, \mathcal{D}\in {[\Gamma]\choose \mu\Gamma},n\in [N]\right\}.
    \end{align}
So cache-node $C_{\gamma}$ totally caches $N{\Gamma-1\choose \mu\Gamma-1}$ subfiles, i.e., $\frac{N{\Gamma-1\choose \mu\Gamma-1}}{{\Gamma\choose \mu\Gamma}}=\frac{NM}{N}=M$ files. Based on the $t$-$(\Gamma,L,K,1)$ design $(\mathcal{X},\mathfrak{B})$, we further divide each subfile into ${L\choose t}$  packets, i.e., $W_{n,\mathcal{D}}=(W^{\mathcal{T}}_{n,\mathcal{D}})_{\mathcal{T}\in {[L]\choose t}}$ for each $n\in [N]$ and $\mathcal{D}\in {[\Gamma]\choose \mu\Gamma}$. Then the content cached by cache-node $C_{\gamma}$ in \eqref{eq-cache-packets} can be written as
   \begin{align}\label{eq-cache-subpackets}
    \mathcal{Z}_{C_\gamma}=\left\{W^{\mathcal{T}}_{n,\mathcal{D}}\ \Big|\ \gamma\in\mathcal{D},\mathcal{D}\in {[\Gamma]\choose \mu\Gamma}, \mathcal{T}\in {[L]\choose t}, n\in [N]\right\}.
    \end{align}
From \eqref{eq-cache-subpackets} and by Definition \ref{defn:three arrays} we can define a ${\Gamma\choose \mu\Gamma}{L\choose t}\times \Gamma$ node-placement array $\mathbf{C}=\left(\mathbf{C}((\mathcal{D},\mathcal{T}),\gamma)\right)_{\mathcal{D}\in {[\Gamma]\choose \mu\Gamma}, \mathcal{T}\in {[L]\choose t},\gamma \in [\Gamma]}$ to represent the subpackets cached by the cache-nodes where
	\begin{align}
		\label{eq-array-node-caching}
		\mathbf{C}((\mathcal{D},\mathcal{T}),\gamma)=\left\{
		\begin{array}{ll}
			* & \ \ \ \hbox{if}\ \ \gamma\in \mathcal{D} \\
			\text{null} & \ \ \ \hbox{otherwise}
		\end{array} \ .
		\right.
	\end{align}
Note that the row index of 	$\mathbf{C}$ is $(\mathcal{D},\mathcal{T})$ corresponding to each packet; the column index is $\gamma$ corresponding to each cache-node.
	From \eqref{eq-cache-packets} and \eqref{eq-array-node-caching}, it can be seen that there are $\mu\Gamma$ stars in each row of $\mathbf{C}$, which means that each packet is stored by $\mu\Gamma$ cache-nodes.

Let us return to the example in Section~\ref{subsect-sketch-1}. In the $(L,r,K,\Gamma,M,N)=(3,3,7,7,1,7)$ coded caching scheme for the $t$-$(\Gamma,L,K,\lambda)=2$-$(7,3,7,1)$ design $(\mathcal{X},\mathfrak{B})$ access topology, by~\eqref{eq-cache-subpackets} and~\eqref{eq-array-node-caching}
 the packets cached by cache-nodes are in \eqref{eq-exa-design-cache-node} and the node-placement array is listed in Table~\ref{tab-caching-node}, respectively.
 \subsection{User-retrieve array}
  Given a $t$-$(\Gamma,L,K,1)$ design $(\mathcal{X},\mathfrak{B})$ access topology, let $U_{\mathcal{B}}$ denote the user who can access the cache-node $C_{\gamma}$ where  $\gamma\in \mathcal{B}$. Then from \eqref{eq-cache-subpackets}   user $U_{\mathcal{B}}$ where $\mathcal{B}\in \mathfrak{B}$ can retrieve the following  packets,
   \begin{align}\label{eq-retrieve-subpackets}
    \mathcal{Z}_{U_{\mathcal{B}}}= \bigcup\limits_{\gamma\in \mathcal{B}}\mathcal{Z}_{C_{\gamma}}=\left\{W^{\mathcal{T}}_{n,\mathcal{D}}\ \Big|\ \mathcal{B}\cap\mathcal{D}\neq \emptyset, \mathcal{D}\in {[\Gamma]\choose \mu\Gamma}, \mathcal{T}\in {[L]\choose t},n\in [N]\right\}.
    \end{align}
From \eqref{eq-retrieve-subpackets} and by Definition \ref{defn:three arrays} we can define a ${\Gamma\choose \mu\Gamma}{L\choose t}\times K$ user-retrieve array $$\mathbf{U}=\left(\mathbf{U}((\mathcal{D},\mathcal{T}),\mathcal{B})\right)_{\mathcal{D}\in {[\Gamma]\choose \mu\Gamma}, \mathcal{T}\in {[L]\choose t},\mathcal{B}\in \mathfrak{B}}$$ to represent the packets retrieved by the users where each entry $\mathbf{U}((\mathcal{D},\mathcal{T}),\mathcal{B})$ can be defined as follows,
\begin{align}
		\label{eq-array-user-retrieve}
		\mathbf{U}((\mathcal{D},\mathcal{T}),\mathcal{B})=\left\{
		\begin{array}{ll}
			* & \ \ \ \hbox{if}\ \ \mathcal{B}\cap \mathcal{D}\neq\emptyset\\
			\text{null} & \ \ \ \hbox{otherwise}
		\end{array} \ .
		\right.
	\end{align}
Each column of $\mathbf{U}$ (representing one user) has exactly ${\Gamma -L\choose \mu\Gamma}\times {L\choose t}$ non-star entries. So there are $$Z=\left({\Gamma\choose \mu\Gamma}-{\Gamma -L\choose \mu\Gamma}\right){L\choose t}$$ stars in each column. Since each user can retrieve $Z$  packets of each file, the local caching gain is $\frac{F-Z}{F}={\Gamma -L\choose \mu\Gamma}/{\Gamma\choose \mu\Gamma} $ which  increases with the access degree $L$.

Let us return to the example in Section~\ref{subsect-sketch-1}.
In the $(L,K,\Gamma,M,N)=(3,7,7,1,7)$ coded caching scheme for the $t$-$(\Gamma,L,K,\lambda)=2$-$(7,3,7,1)$ design $(\mathcal{X},\mathfrak{B})$ access topology, by~\eqref{eq-retrieve-subpackets} and~\eqref{eq-array-user-retrieve},  the packets retrieved by users are in \eqref{eq-exa-design-cache-user}  and the user-retrieve array is listed in Table \ref{tab-user-retrieve}, respectively.

\subsection{User-delivery array}
\label{subsection-design-delivery-array}
Based on $\mathbf{U}$, the ${\Gamma\choose \mu\Gamma}{L\choose t}\times K$ user-delivery array $\mathbf{Q}=\left(\mathbf{Q}((\mathcal{D},\mathcal{T}),\mathcal{B}) \right)_{\mathcal{D}\in {[\Gamma]\choose \mu\Gamma}, \mathcal{T}\in {[L]\choose t}, \mathcal{B}\in \mathfrak{B}}$ to represent the transmitted packets to the users can be defined as follows,
	\begin{align}
		\label{eq-array-user-PDA}
		\mathbf{Q}((\mathcal{D},\mathcal{T}),\mathcal{B})=\left\{
		\begin{array}{ll}
			\mathcal{D}\cup\mathcal{B}(\mathcal{T}) & \ \ \ \hbox{if}\ \ \mathcal{D}\cap\mathcal{B}=\emptyset \\
			* & \ \ \ \hbox{otherwise}
		\end{array},
		\right.
	\end{align}
where $\mathcal{B}(\mathcal{T})$ represents the elements in $\mathcal{B}$ with index in $\mathcal{T}$.\footnote{\label{foot:ex1 user-delveliry}Let us return to the example in Section~\ref{subsect-sketch-1} again. In the $(L,K,\Gamma,M,N)=(3,7,7,1,7)$ coded caching scheme for the $t$-$(\Gamma,L,K,\lambda)=2$-$(7,3,7,1)$ design $(\mathcal{X},\mathfrak{B})$ access topology, by~\eqref{eq-array-user-PDA}, the user-delivery array is listed in Table~\ref{tab-user-delivery}.}

It can be seen that the resulting array $\mathbf{Q}$ only contains the $(\mu\Gamma+ t)$-subsets of $[\Gamma]$ which indicate the broadcasted messages. Let us count the number of these different subsets occurring in $\mathbf{Q}$. For any $(t+\mu\Gamma)$-subset of $\mathcal{X}$, denoted by $\mathcal{S}$, where  $\mathcal{S}$ is a subset of some block $\mathcal{B}$ in $\mathfrak{B}$, one can see that $\mathcal{S}$ cannot be an entry in the array $\mathbf{Q}$. Otherwise, suppose there is an entry $\mathcal{D}\cup\mathcal{B}'(\mathcal{T}')=\mathcal{S}$ in the array $\mathbf{Q}$. Then we have $\mathcal{B}'(\mathcal{T}')\subseteq \mathcal{B}'\cap \mathcal{B}$. Recall that from the property of $t$-design with $\lambda=1$, $\mathcal{B}=\mathcal{B}'$ always holds. So we have $\mathcal{D}\subseteq\mathcal{S}\subseteq\mathcal{B}$ which implies that this entry is a star from \eqref{eq-array-user-PDA}; by contradiction, we can prove that $\mathcal{S}$ cannot be an entry in the array $\mathbf{Q}$.

As a result, there are at most ${\Gamma \choose t+\mu\Gamma}-K{L\choose t+\mu\Gamma}$ different $(t+\mu\Gamma)$-subset of $\mathcal{X}$ occurring in $\mathbf{Q}$.
Thus the number of multicast messages is at most ${\Gamma \choose t+\mu\Gamma}-K{L\choose t+\mu\Gamma}$.

Let us then show that $\mathbf{Q}$ satisfies Condition C$3$ of Definition \ref{def-PDA}.

Consider Condition C3-a) first. Assume that there exists a  set  occurring in the same row, denoted by  $
\mathcal{S}=\mathbf{Q}((\mathcal{D},\mathcal{T}),\mathcal{B}_1)
=\mathbf{Q}((\mathcal{D},\mathcal{T}),\mathcal{B}_2) $,
 where $\mathcal{D}\in{[\Gamma]\choose \mu\Gamma}$, $\mathcal{T}\in{[L]\choose t}$, $\mathcal{B}_1$, $\mathcal{B}_2\in \mathfrak{B}$ and $\mathcal{B}_1\neq \mathcal{B}_2$. From \eqref{eq-array-user-PDA}, we have
\begin{align*}
\mathcal{D}\cap\mathcal{B}_1=\mathcal{D}\cap\mathcal{B}_2=\emptyset,\ \ \mathcal{D}\cup\mathcal{B}_1(\mathcal{T})=
\mathcal{D}\cup\mathcal{B}_1(\mathcal{T}).
\end{align*}This implies that $\mathcal{B}_1(\mathcal{T})=\mathcal{B}_2(\mathcal{T})$ which is impossible since the parameter $\lambda=1$ in our used design. Similarly assume that there exists a set  occurring in the same column, denoted by
$ \mathcal{C}=
\mathbf{Q}((\mathcal{D}_1,\mathcal{T}_1),\mathcal{B})
=\mathbf{Q}((\mathcal{D}_2,\mathcal{T}_2),\mathcal{B}),
$ where $\mathcal{D}_1,\mathcal{D}_2\in{[\Gamma]\choose \mu\Gamma}$, $\mathcal{T}_1,\mathcal{T}_2\in{[L]\choose t}$, $\mathcal{B}\in \mathfrak{B}$ and $(\mathcal{D}_1,\mathcal{T}_1)\neq (\mathcal{D}_2,\mathcal{T}_2)$. From \eqref{eq-array-user-PDA}, we have
\begin{align*}
\mathcal{D}_1\cap\mathcal{B}=\mathcal{D}_2\cap\mathcal{B}=\emptyset,\ \ \mathcal{D}_1\cup\mathcal{B}(\mathcal{T}_1)=
\mathcal{D}_2\cup\mathcal{B}(\mathcal{T}_2).
\end{align*}This implies that $\mathcal{D}_1=\mathcal{D}_2$. Then we have $\mathcal{B}(\mathcal{T}_1)=\mathcal{B}(\mathcal{T}_2)$, i.e., $\mathcal{T}_1=\mathcal{T}_2$. So $(\mathcal{D}_1,\mathcal{T}_1)= (\mathcal{D}_2,\mathcal{T}_2)$ which is contradiction with our hypothesis $(\mathcal{D}_1,\mathcal{T}_1)\neq (\mathcal{D}_2,\mathcal{T}_2)$. From the above introduction, we have that each subset of $\mathcal{X}$ occurring in $\mathbf{P}$ occurs at most once in each row and each column, i.e., Condition C3-a) holds.

Next consider Condition C3-b). Assume that there exist two
 different entries in $\mathbf{Q}$ which are filled by the same set,
\begin{align}\label{eq-C3-proof}
\mathbf{Q}((\mathcal{T},\mathcal{D}),\mathcal{B})=\mathbf{Q}((\mathcal{T}',\mathcal{D}'),\mathcal{B}')
=\mathcal{B}(\mathcal{T}) \cup\mathcal{D}=\mathcal{B}'(\mathcal{T}')\cup\mathcal{D}'.
\end{align}
It can be seen that  $\mathcal{B}\neq \mathcal{B}'$ and $\mathcal{D}\neq \mathcal{D}'$,  since $\mathbf{Q}$ satisfies Condition C3-a).
Then for any $x\in \mathcal{D}\setminus \mathcal{D}'$ and $x'\in \mathcal{D}'\setminus\mathcal{D}$, we have
$$x\in \mathcal{B}'(\mathcal{T}') \ \ \ \text{and}\ \ \ x'\in \mathcal{B}(\mathcal{T})  ,$$
 which implies that
$\mathcal{B}'\bigcap \mathcal{D}\neq\emptyset$ and $\mathcal{B}\bigcap \mathcal{D}'\neq\emptyset$. From \eqref{eq-array-user-PDA} we have  $\mathbf{Q}((\mathcal{T},\mathcal{D}),\mathcal{B}')=\mathbf{Q}((\mathcal{T}',\mathcal{D}'),\mathcal{B})=*$, implying that Condition C3-b) is satisfied.

As a result, $\mathbf{Q}$ is a $(K,F,Z,S)$ PDA, where
\begin{subequations}
\begin{align}
K&={\Gamma\choose t}/{L\choose t},\ \ F={\Gamma\choose \mu\Gamma}{L\choose t},\ \ Z=\left({\Gamma\choose \mu\Gamma}-{\Gamma -L\choose \mu\Gamma}\right){L\choose t},\label{eq-PDA-parameters-K-F-Z}\\
S& \leq {\Gamma \choose t+\mu\Gamma}-K{L\choose t+\mu\Gamma}.\label{eq-PDA-parameters-S}
\end{align}
\end{subequations}
Then by Lemma \ref{le-Fundamental}, we have a $(L,r,K,\Gamma, M,N)$ multiaccess coded caching scheme based on a $t$-$(\Gamma,L,K,r,1)$   design with transmission load
\begin{align*}
R& \leq    \frac{{\Gamma \choose t+\mu\Gamma}-K{L\choose t+\mu\Gamma}}{{\Gamma\choose \mu\Gamma}{L\choose t}}.
\end{align*}
We should point out that when $\mu\Gamma>\Gamma-L$, we have ${\Gamma -L\choose \mu\Gamma}=0$. So the parameter
$Z=\left({\Gamma\choose \mu\Gamma}-{\Gamma -L\choose \mu\Gamma}\right){L\choose t}={\Gamma\choose \mu\Gamma}{L\choose t}=F$ in \eqref{eq-PDA-parameters-K-F-Z}. This implies that each user can retrieve all the packets of any file. So in this paper we only need to study the case $\mu\Gamma\leq\Gamma-L$.

\subsection{Further reduce the amount of transmission}
\label{sub-perform}
In addition, we can further reduce the transmissions by using the following observation.
\begin{proposition}\rm
\label{pro-1}
In $\mathbf{Q}$ generated by \eqref{eq-array-user-PDA}, for any column labelled by $\mathcal{B}\in \mathfrak{B}$ and   any
 $(t+\mu\Gamma)$-subset $\mathcal{S}$ of $[\Gamma]$ with $t+1\leq|\mathcal{S}\cap \mathcal{B}|<t+\mu\Gamma$, the following   always holds:
\begin{itemize}
\item
$\mathbf{Q}((\mathcal{D},\mathcal{T}),\mathcal{B})=*$, where
 $\mathcal{D}\in {[\Gamma]\choose \mu\Gamma}$, $\mathcal{T}\in {[L]\choose t}$, $\mathcal{B}'\in \mathfrak{B}$, and   $\mathbf{Q}((\mathcal{D},\mathcal{T}),\mathcal{B}')=\mathcal{S}$.
\end{itemize}
\hfill $\square$
\end{proposition}
\begin{proof}
Since  $\mathbf{Q}((\mathcal{D},\mathcal{T}),\mathcal{B}')=\mathcal{S}$, $\mathcal{S}$ is composed of all elements in $\mathcal{D}$ and $t$ elements in $\mathcal{B}'$.
 Since  $t+1\leq|\mathcal{S}\cap \mathcal{B}|<t+\mu\Gamma$, there must be at least one element in $\mathcal{D}$ which is also in $\mathcal{B}$, and so $\mathbf{Q}((\mathcal{D},\mathcal{T}),\mathcal{B})=*$.
\end{proof}
Given $\mathcal{B}$, there are exactly
\begin{align}\label{eq-redundancy}
S'=\sum\limits_{i\in [\mu\Gamma-1]}{L\choose t+i}{\Gamma-L\choose \mu\Gamma -i}
\end{align} $(t+\mu\Gamma)$-subsets $\mathcal{S}$ of $[\Gamma]$ satisfying $t+1\leq |\mathcal{S}\cap \mathcal{B}|<t+\mu\Gamma$.   By Proposition \ref{pro-1}, the number of the   multicast messages (i.e., XOR of packets in each time slot) which are unknown to  user $U_{\mathcal{B}}$ is at most $S-S'$. Thus, using a $[S,S-S']$  MDS   code, the server only needs to send $S-S'$ coded packets. Then the transmission load is
\begin{eqnarray}
\label{eq-load2}
\begin{split}
 R_{\text{Th1}} &=\frac{{\Gamma \choose t+\mu\Gamma}-\sum^{\mu\Gamma-1}_{i=1}{L\choose t+i}{\Gamma-L\choose \mu\Gamma -i}-K{L\choose t+\mu\Gamma}}{{L\choose t}{\Gamma\choose \mu\Gamma}}&\\[0.2cm]
&=\frac{{\Gamma \choose t+\mu\Gamma}}{{\Gamma\choose \mu\Gamma}{L\choose t}}-
\frac{\sum^{\mu\Gamma-1}_{i=1}{L\choose t+i}{\Gamma-L\choose \mu\Gamma -i}}{{L\choose t}{\Gamma\choose \mu\Gamma}}-
\frac{K{L\choose t+\mu\Gamma}}{{L\choose t}{\Gamma\choose \mu\Gamma}}&
\end{split}
\end{eqnarray}
which coincides with Theorem~\ref{th-PDA-SS}..

\section{Multiaccess coded caching scheme for the $t$-GDD access topology: Proof of Theorem~\ref{th-GDD}}
\label{sec-Scheme-GDD}
In this section, we focus on the MACC system with the $t$-GDD access topology and propose a unified construction on the coded caching scheme given a $t$.
 Then we show that by extending the proposed MACC scheme to the shared-link model, the resulting scheme covers the shared-link caching scheme in \cite{CWZW} as a special case.

For any positive integers $m$, $q\geq 2$, $s$, $t$ and $L$ satisfying $1\leq t\leq L\leq s\leq m$, assume that there exists a $t$-$(m,q,L,1)$ GDD $(\mathcal{X},\mathfrak{G},\mathfrak{B})$ and an OA$(m,q,s)$  $\mathbf{A}=(\mathbf{A}(j,u))_{j\in [q^s],u\in [m]}$. Without loss of generality, we can let  $\mathcal{X}=\{(u,v)\ |\ u\in [m],v\in [q]\}$, $\mathfrak{G}=\{\mathcal{G}_1,\mathcal{G}_1,\ldots,
\mathcal{G}_{m}\}$ and $\mathfrak{B}=\{\mathcal{B}_1,\mathcal{B}_2,\ldots,\mathcal{B}_K\}$ where
\begin{subequations}
\label{eq-GDD-groups-blocks}
\begin{align}
\mathcal{G}_u&=\{(u,1),(u,2),\ldots,(u,q)\}, \  \forall  u\in [m]\label{eq-GDD-groups-OA}\\
\mathcal{B}_k&=\{(u_{k,1},v_{k,1}),\ldots,(u_{k,L},v_{k,L})\},\ 1\leq u_{k,1}<u_{k,2}<\cdots<u_{k,L}\leq m, \nonumber\\
 &\ \ \ \ \ \ \ \ \ \ \ \ \ \ \ \ \ \ \ \ \ \ \ \ \ \ \ \ \ \ \ \ \ \ \ \ \ \ \ \ 1\leq v_{k,1},v_{k,2},\ldots,v_{k,L}\leq q.\label{eq-GDD-blocks-OA}
\end{align}
\end{subequations}
From \eqref{eq-GDD-blocks-number} we have $K=|\mathfrak{B}|={m\choose t}q^t/{L\choose t}$.
 Then we can obtain a $(L,K=|\mathfrak{B}|,\Gamma=mq,M,N)$ coded caching scheme with $M/N=1/q$ under $(\mathcal{X},\mathfrak{G},\mathfrak{B})$ access topology and OA$(m,q,s)$  placement strategy  by constructing the following node-placement, user-retrieval  and user-delivery arrays.

\subsection{Node-placement array}
We will use the rows $\mathbf{A}(j,\cdot)$ where $j\in [q^s]$ of OA$(m,q,s)$ $\mathbf{A}=(\mathbf{A}(j,u))_{j\in [q^s],u\in [m]}$ as the packet labels and place the packets on each cache-node according to the entries in $\mathbf{A}$. Specifically, the server divides each file into $q^s$ subfiles with equal size, i.e., for each $n\in [N]$, $W_n=(W_{n,j})_{j\in [q^s]}$.  For simplicity, the cache-nodes are denoted by $C_{u,v}$ where $u\in [m]$ and $v\in [q]$, i.e., $C_{\gamma}=C_{u,v}$ if and only if $u=\lceil \gamma/q\rceil$ and $v=<\gamma>_q$.
 Each cache-node $C_{u,v}$ caches the following subfiles
    \begin{align}\label{eq-cache-packets-OA}
    \mathcal{Z}_{C_{u,v}}=\left\{W_{n,j}\ \Big|\ \mathbf{A}(j,u)=v,j\in [q^s], n\in [N]\right\}.
    \end{align}
Based on the $t$-$(m,q,L,1)$ GDD $(\mathcal{X},\mathfrak{G},\mathfrak{B})$, we further divide each subfile into ${L\choose t}$ packets, i.e., $W_{n,j}=(W^{\mathcal{T}}_{n,j})_{\mathcal{T}\in {[L]\choose t}}$ for each $n\in [N]$ and $j\in [q^s]$. Then the packets cached by the cache-node $C_{u,v}$ in \eqref{eq-cache-packets-OA} can be written as
   \begin{align}\label{eq-cache-subpackets-GDD}
    \mathcal{Z}_{C_{u,v}}=\left\{W^{\mathcal{T}}_{n, j }\ \Big|\ \mathbf{A}(j,u)=v,j\in [q^s], \mathcal{T}\in {[L]\choose t}, n\in [N]\right\}.
    \end{align}
From \eqref{eq-cache-subpackets-GDD} we can define a ${L\choose t}q^{s}\times mq$ node-placement array $\mathbf{C}=\left(\mathbf{C}((j,{\mathcal{T}}),(u,v))\right)$, where $\mathcal{T}\in {[L]\choose t}$, $j\in [q^s]$ and $u\in [m],v\in [q]$, to represent the packets cached by the cache-nodes where
\begin{align}
\label{eq-array-node-caching-GDD}
\mathbf{C}((j,{\mathcal{T}}),(u,v))=\left\{
\begin{array}{ll}
* & \ \ \ \hbox{if}\ \ \mathbf{A}(j,u)=v \\
\text{null} & \ \ \ \hbox{otherwise}
\end{array} \ .
\right.
\end{align}
Note that the row index of $\mathbf{C}$ is $(j,{\mathcal{T}})$ corresponding to each packet; the column index of $\mathbf{C}$ is $(u,v)$ corresponding to each cache-node.
From \eqref{eq-array-node-caching-GDD}, each column of $\mathbf{C}$ has exactly ${L\choose t}q^{s-1}$ stars which implies that each cache-node caches $Nq^{s-1}{L\choose t}$ packets, i.e., $\frac{N}{q}=\frac{NM}{N}=M$ files; there are $m$ stars in each row, which means that each packet stored by $m$ cache-nodes.

Let us return to the example in Section~\ref{subsect-sketch-2}. In the $(L,K,\Gamma,M,N)=(2,12,6,3,6)$ coded caching scheme with the
OA$(m=3,q=2,s=2)$ placement for the
$2$-$(3,2,2,1)$ GDD access topology, by \eqref{eq-cache-subpackets-GDD} and~\eqref{eq-array-node-caching-GDD}, the packets cached by cache-nodes are in \eqref{eq-cache-subpackets-exam} and the node-placement array is   in Table \ref{tab-caching-node-OA}, respectively.
\subsection{User-retrieve array}
We treat each block of $\mathfrak{B}$ as the index of a user. For each integer $k\in [K]$, from \eqref{eq-GDD-blocks-OA} we use $\psi(\mathcal{B}_k)$ to represent the set consisting of the first coordinate of each point in $\mathcal{B}_k$ and use $\psi(\mathcal{B}_k)$ to represent the vector generated by the second coordinate of each point in $\mathcal{B}_k$, i.e.,
\begin{align*}
\psi(\mathcal{B}_k)=\{u_{k,1},u_{k,2},\ldots,u_{k,L}\},\ \
\psi(\mathcal{B}_k)=(v_{k,1},v_{k,2},\ldots,v_{k,L}).
\end{align*} Then from \eqref{eq-cache-subpackets-GDD} user $U_{\mathcal{B}_k}$ where $k\in [K]$ can retrieve the following subpackets.
\begin{align}\label{eq-retrieve-subpackets-GDD}
    \mathcal{Z}_{U_{\mathcal{B}_k}}= \bigcup\limits_{i\in [L]}\mathcal{Z}_{C_{u_i,v_i}}=\left\{W^{\mathcal{T}}_{n,j}\ \Big|\ \text{d}\left(\mathbf{A}(j,\psi(\mathcal{B}_k)),\varphi(\mathcal{B}_k)\right)<L,j\in [q^{s}], \mathcal{T}\in {[L]\choose t}, n\in [N]\right\}.
\end{align} Note that $\text{d}\left(\mathbf{A}(j,\psi(\mathcal{B}_k)),\varphi(\mathcal{B}_k)\right)$ is the hamming distance between the vectors $\mathbf{A}(j,\psi(\mathcal{B}_k))$ and $\varphi(\mathcal{B}_k)$, i.e., the number of coordinates in which $\mathbf{A}(j,\psi(\mathcal{B}_k))$ and $\varphi(\mathcal{B}_k)$ differ. From \eqref{eq-retrieve-subpackets-GDD} we can define a ${L\choose t}q^s\times K$ user-retrieve array $\mathbf{U}=\left(\mathbf{U}((j,{\mathcal{T}}), \mathcal{B}_k)\right)_{j\in [q^s], \mathcal{T}\in {[L]\choose t}, k\in [K]}$ to represent the packets retrieved by the users where each entry $\mathbf{U}((j,{\mathcal{T}}),\mathcal{B}_k)$ can be defined as follows.
\begin{align}
\label{eq-array-user-retrieve-GDD}
\mathbf{U}((j,{\mathcal{T}}),\mathcal{B}_k)=\left\{
\begin{array}{ll}
* & \ \ \ \hbox{if}\ \ \text{d}\left(\mathbf{A}(j,\psi(\mathcal{B}_k)),\varphi(\mathcal{B}_k)\right)<L\\
\text{null} & \ \ \ \hbox{otherwise}
\end{array}
\right.
\end{align}

For example, in the $(L,K,\Gamma,M,N)=(2,12,6,3,6)$ coded caching scheme under OA$(m=3,q=2,s=2)$ placement and
$2$-$(3,2,2,1)$ GDD access topology in \eqref{eq-2-3-2-2-1-GDD} of Subsection \ref{subsect-sketch-2}, from \eqref{eq-retrieve-subpackets-GDD}, the packets retrieved by users can be obtained in \eqref{eq-retrieve-subpackets-exam}. From \eqref{eq-array-user-retrieve-GDD} the user-retrieve array is obtained in Table \ref{tab-user-retrieve-GDD}.

\subsection{User-delivery array}
\label{subsec-UDA-GDD}
Using the $\mathbf{U}$ generated in \eqref{eq-array-user-retrieve-GDD}, the ${L\choose t}q^s\times {m\choose t}q^t/{L\choose t}$ user-delivery array $$\mathbf{Q}=\left(\mathbf{Q}((j,{\mathcal{T}}), \mathcal{B}_k)\right)_{j\in [q^s],\mathcal{T}\in{[L]\choose t},k\in [K]}$$ to represent  the transmitted packets to the users can be defined as follows.
\begin{align}
\label{eq-array-user-PDA-GDD}
\mathbf{Q}((j,{\mathcal{T}}), \mathcal{B}_k)=\left\{
\begin{array}{ll}
({\bf e},n_{\bf e}) & \textrm{if}~\text{d}\left(\mathbf{A}(j,\psi(\mathcal{B}_k)),\varphi(\mathcal{B}_k)\right)=L, \\
* & \textrm{otherwise},
\end{array}
\right.
\end{align}
where ${\bf e}=(e_1,e_{2},\ldots,e_{m})\in[q]^m$ such that for each $i\in [m]$ \begin{eqnarray}
\label{eq-putting integer}
e_i=\left\{
\begin{array}{ll}
\varphi(\mathcal{B}_k)(h) & \textrm{if}\ i=\psi(\mathcal{B}_k)(h)\in \mathcal{T}\ \text{for some}\ h\in [t],\\[0.2cm]
\mathbf{A}(j,i) & \textrm{otherwise,} \end{array}
\right.
\end{eqnarray}
$n_{\bf e}$ is the order of occurrence of the vector ${\bf e}$ in column labelled by $\mathcal{B}_k$ and starts from $1$.

For example, in the $(L,K,\Gamma,M,N)=(2,12,6,3,6)$ coded caching scheme under OA$(m=3,q=2,s=2)$ placement and
$2$-$(3,2,2,1)$ GDD access topology of Subsection \ref{subsect-sketch-2}, from \eqref{eq-array-user-PDA-GDD} and \eqref{eq-putting integer}, the user-delivery array is obtained in Table \ref{tab-user-delivery-GDD}. Since each vector occurs at most once in each column of Table \ref{tab-user-delivery-GDD}, we omit the occurrence number of each vector.

Now let us show that the obtained array $\mathbf{Q}$ satisfies Condition C3 of Definition \ref{def-PDA}. For any two different entries say $\mathbf{Q}((j,{\mathcal{T}}), \mathcal{B}_k)$ and $\mathbf{Q}((j',\mathcal{T}'), \mathcal{B}_{k'})$ where $j$, $j'\in [q^s]$, $\mathcal{T}$, $\mathcal{T}'\in {[L]\choose t}$ and $k,k'\in [K]$, assume that
\begin{align}\label{eq-equal}
\mathbf{Q}((j,{\mathcal{T}}), \mathcal{B}_k)=\mathbf{Q}((j',\mathcal{T}'), \mathcal{B}_{k'})=(e_1,e_2,\ldots,e_{m},n).
\end{align}From \eqref{eq-GDD-groups-blocks}, we denote the blocks $\mathcal{B}_k$ and $\mathcal{B}_{k'}$ as follows
\begin{align*}
\mathcal{B}_k=\{(u_{k,1},v_{k,1}),\ldots,(u_{k,L},v_{k,L})\}\ \ \ \
\mathcal{B}_{k'}=\{(u_{k',1},v_{k',1}),\ldots,(u_{k',L},v_{k',L})\}
\end{align*} where $1\leq u_{k,1}<u_{k,2}<\cdots<u_{k,L}\leq m$, $1\leq u_{k',1}<u_{k',2}<\cdots<u_{k',L}\leq m$, $1\leq v_{k,1},v_{k,2},\ldots,v_{k,L}\leq q$ and $1\leq v_{k',1},v_{k',2},\ldots,v_{k',L}\leq q$.

\begin{itemize}
\item Let us consider condition C3-a) of Definition \ref{def-PDA}. From \eqref{eq-array-user-PDA-GDD} we have $\mathcal{B}_k\neq \mathcal{B}_{k'}$, i.e., $k\neq k'$, since all the vectors are different in each column. So each vector occurs at most once in each column. Assume that $(e_1,e_2,\ldots,e_{m})$ occurs twice in a row $(j,\mathcal{T})$, i.e., $(j,\mathcal{T})=(j',\mathcal{T}')$. Let us consider the subsets $\psi(\mathcal{B}_k)$, $\psi(\mathcal{B}_{k'})$ and the vectors $\varphi(\mathcal{B}_k)$, $\varphi(\mathcal{B}_{k'})$ respectively. If $\psi(\mathcal{B}_k)(\mathcal{T})=\psi(\mathcal{B}_{k'})(\mathcal{T})$ we have $\varphi(\mathcal{B}_k)\neq\varphi(\mathcal{B}_{k'})$ by the definition of $t$-GDD. Then from \eqref{eq-array-user-PDA-GDD} we have $\mathbf{Q}((j,{\mathcal{T}}), \mathcal{B}_k)\neq\mathbf{Q}((j',\mathcal{T}'), \mathcal{B}_{k'})$ which contradicts our hypothesis. When $\psi(\mathcal{B}_k)(\mathcal{T})\neq\psi(\mathcal{B}_{k'})(\mathcal{T})$ without loss of generality there must be two different integers say $u_{k,i}$ and $u_{k',i'}\in [m]$ where $i,i'\in[L]$ such that
\begin{align*}
u_{k,i}\in \psi(\mathcal{B}_k)(\mathcal{T})
\setminus\psi(\mathcal{B}_{k'})(\mathcal{T})\ \ \ \text{and}\ \ \
u_{k',i'}\in \psi(\mathcal{B}_{k'})(\mathcal{T})
\setminus\psi(\mathcal{B}_k)(\mathcal{T})
\end{align*}
always hold. From \eqref{eq-array-user-retrieve-GDD} we have
\begin{align*}
\mathbf{A}(j,u_{k',i'})=e_{u_{k',i'}}=v_{k',i'}\ \ \ \ \ \text{and}\ \ \ \ \  \mathbf{A}(j,u_{k,i})=e_{u_{k,i}}=v_{k,i}
\end{align*} which implies that $\mathbf{Q}((j,{\mathcal{T}}), \mathcal{B}_k)=\mathbf{Q}((j',\mathcal{T}'), \mathcal{B}_{k'})=*$, a contradiction to our hypothesis. So Condition C3-a) holds.
\item Let us consider C3-b). That is \eqref{eq-equal} holds for the case $(j,\mathcal{T})\neq(j',\mathcal{T}')$ and $\mathcal{B}_k\neq \mathcal{B}_{k'}$ (i.e., $k\neq k'$). Recall that every $t$-subset of points from $t$ different groups belongs to exactly $\lambda=1$ block, i.e, the third property of a $t$-GDD,
\begin{align*}
(\psi(\mathcal{B}_k)(\mathcal{T}),\varphi(\mathcal{B}_k)(\mathcal{T}))\neq (\psi(\mathcal{B}_{k'})(\mathcal{T}'),\varphi(\mathcal{B}_{k'})(\mathcal{T}'))
\end{align*}
always holds for any $\mathcal{T}$, $\mathcal{T}'\in {[L]\choose t}$. When $\psi(\mathcal{B}_k)(\mathcal{T})= \psi(\mathcal{B}_{k'})(\mathcal{T}')$ we have $\varphi(\mathcal{B}_k)(\mathcal{T})\neq \varphi(\mathcal{B}_{k'})(\mathcal{T}')$. Then from \eqref{eq-putting integer} we have $
\mathbf{Q}((j,{\mathcal{T}}), \mathcal{B}_k)\neq\mathbf{Q}((j',\mathcal{T}'),\mathcal{B}_{k'})$, a contradiction to our hypothesis in \eqref{eq-equal}. When $\psi(\mathcal{B}_k)(\mathcal{T})\neq \psi(\mathcal{B}_{k'})(\mathcal{T}')$, similar to proof of Condition C3-a), we also have that $\mathbf{Q}((j,{\mathcal{T}}), \mathcal{B}_k)=\mathbf{Q}((j',\mathcal{T}'), \mathcal{B}_{k'})=*$, a contradiction to our hypothesis.
\end{itemize}
From the above discussion, $\mathbf{Q}$ generated by \eqref{eq-array-user-PDA-GDD} satisfies Condition C3 of a PDA. Now let us compute the parameters $K$, $F$, $Z$ and $S$ respectively. First we have $K={m\choose t}q^t/{L\choose t}$, $F=q^s{L\choose t}$ and $Z=\left(q^s-(q-1)^Lq^{s-L}\right){L\choose t}$.

Finally let us consider the value of $S$. First there are at most $q^m$ different vectors ${\bf e}$ with length $m$. In addition, from \eqref{eq-array-user-PDA-GDD} and \eqref{eq-putting integer} it is easy to count that each vector ${\bf e}$ occurs at most $(q-1)^{t}$ times in a column. So the value of $S$ is at most $(q-1)^tq^m$. In fact it is not easy to count the exact value of $S$ for any parameters $L$, $t$, $s$ and $m$. However when $L=t$, we can check that \eqref{eq-array-user-PDA-GDD} is exactly the formula (2) in
\cite[Construction 1]{CWZW}. Then the proof is completed.

\section{The case $\lambda>1$}
\label{section-lambda>1-th1}
Using the same placement strategy in \eqref{eq-cache-subpackets} and \eqref{eq-retrieve-subpackets}, we can also obtain the scheme under any $t$-design accessing topology with index $\lambda>1$ by just replacing the $(t+\mu\Gamma)$-subset $\mathcal{D}\cup\mathcal{B}(\mathcal{T})$, which is defined in \eqref{eq-array-user-PDA-lambda} for the non-star entry at the row labelled by $(\mathcal{D},\mathcal{T})$ and the column labelled by $\mathcal{B}$, by a pair $(\mathcal{D}\cup\mathcal{B}(\mathcal{T}),
n_{\mathcal{D},\mathcal{B}(\mathcal{T})})$ where $n_{\mathcal{D},\mathcal{B}(\mathcal{T})}$ is the occurrence order of pair $(\mathcal{D},\mathcal{B}(\mathcal{T}))$ from top to bottom and left to right in $\mathbf{Q}$. That is using the same constructions of node-placement array in \eqref{eq-array-node-caching}, user-retrieve array in \eqref{eq-array-user-retrieve} and changing construction of user-delivery array $\mathbf{Q}$ in \eqref{eq-array-user-PDA} as follows.

\begin{align}
		\label{eq-array-user-PDA-lambda}
		\mathbf{Q}((\mathcal{D},\mathcal{T}),\mathcal{B})=\left\{
		\begin{array}{ll}
			(\mathcal{D}\cup\mathcal{B}(\mathcal{T}),
n_{\mathcal{D},\mathcal{B}(\mathcal{T})}), & \ \ \ \hbox{if}\ \ \mathcal{D}\cap\mathcal{B}=\emptyset \\
			* & \ \ \ \hbox{otherwise}
		\end{array}
		\right.
	\end{align}
For instance, we take the node placement array $\mathbf{C}$ in Table \ref{tab-caching-node} of Subsection \ref{subsect-sketch-1} to illustrate our rule defined in \eqref{eq-array-user-PDA-lambda}. When we use a $2$-$(7,4,7,2)$ design $(\mathcal{X},\mathfrak{B})$ where
$$\mathfrak{B}=\{\{1,2,3,5\},\{2,3,4,6\},\{3,4,5,7\},\{4,5,6,1\},\{5,6,7,2\},\{6,7,1,3\},\{7,1,2,4\}\},$$
as the accessing topology, and from \eqref{eq-array-user-PDA} we have a $42\times 7$ user-retrieve array $\mathbf{U}$ listed in Table \ref{tab-user-retrieve-lambda} to represent the packets retrieved by the users.
\begin{table}[!htbp]
\center
\caption{User-retrieve array $\mathbf{U}$.
\label{tab-user-retrieve-lambda}}
\renewcommand\arraystretch{0.5}
\begin{tabular}{|c|ccccccc|}
Subpacket labels&\multicolumn{7}{|c|}{User set $\mathfrak{B}$} \\ \hline
$\mathcal{D},\mathcal{T}$ & $U_{1235}$ & $U_{2346}$ & $U_{3457}$ & $U_{1456}$ & $U_{2567}$ & $U_{1367}$ & $U_{1247}$ \\ \hline
$\{1\},\{1,2\}$ & * &  &  & * &  & * & * \\
$\{2\},\{1,2\}$ & * & * &  &  & * &  & * \\
$\{3\},\{1,2\}$ & * & * & * &  &  & * &  \\
$\{4\},\{1,2\}$ &  & * & * & * &  &  & * \\
$\{5\},\{1,2\}$ & * &  & * & * & * &  &  \\
$\{6\},\{1,2\}$ &  & * &  & * & * & * &  \\
$\{7\},\{1,2\}$ &  &  & * &  & * & * & * \\\hline
$\{1\},\{1,3\}$ & * &  &  & * &  & * & * \\
$\{2\},\{1,3\}$ & * & * &  &  & * &  & * \\
$\{3\},\{1,3\}$ & * & * & * &  &  & * &  \\
$\{4\},\{1,3\}$ &  & * & * & * &  &  & * \\
$\{5\},\{1,3\}$ & * &  & * & * & * &  &  \\
$\{6\},\{1,3\}$ &  & * &  & * & * & * &  \\
$\{7\},\{1,3\}$ &  &  & * &  & * & * & * \\ \hline
$\{1\},\{1,4\}$ & * &  &  & * &  & * & * \\
$\{2\},\{1,4\}$ & * & * &  &  & * &  & * \\
$\{3\},\{1,4\}$ & * & * & * &  &  & * &  \\
$\{4\},\{1,4\}$ &  & * & * & * &  &  & * \\
$\{5\},\{1,4\}$ & * &  & * & * & * &  &  \\
$\{6\},\{1,4\}$ &  & * &  & * & * & * &  \\
$\{7\},\{1,4\}$ &  &  & * &  & * & * & * \\\hline
$\{1\},\{2,3\}$ & * &  &  & * &  & * & * \\
$\{2\},\{2,3\}$ & * & * &  &  & * &  & * \\
$\{3\},\{2,3\}$ & * & * & * &  &  & * &  \\
$\{4\},\{2,3\}$ &  & * & * & * &  &  & * \\
$\{5\},\{2,3\}$ & * &  & * & * & * &  &  \\
$\{6\},\{2,3\}$ &  & * &  & * & * & * &  \\
$\{7\},\{2,3\}$ &  &  & * &  & * & * & * \\ \hline
$\{1\},\{2,4\}$ & * &  &  & * &  & * & * \\
$\{2\},\{2,4\}$ & * & * &  &  & * &  & * \\
$\{3\},\{2,4\}$ & * & * & * &  &  & * &  \\
$\{4\},\{2,4\}$ &  & * & * & * &  &  & * \\
$\{5\},\{2,4\}$ & * &  & * & * & * &  &  \\
$\{6\},\{2,4\}$ &  & * &  & * & * & * &  \\
$\{7\},\{2,4\}$ &  &  & * &  & * & * & * \\ \hline
$\{1\},\{3,4\}$ & * &  &  & * &  & * & * \\
$\{2\},\{3,4\}$ & * & * &  &  & * &  & * \\
$\{3\},\{3,4\}$ & * & * & * &  &  & * &  \\
$\{4\},\{3,4\}$ &  & * & * & * &  &  & * \\
$\{5\},\{3,4\}$ & * &  & * & * & * &  &  \\
$\{6\},\{3,4\}$ &  & * &  & * & * & * &  \\
$\{7\},\{3,4\}$ &  &  & * &  & * & * & * \\ \hline
\end{tabular}
\end{table}
Then from \eqref{eq-array-user-PDA-lambda} the $42\times 7$ user-delivery array $\mathbf{Q}$ which is listed in  Table \ref{tab-user-delivery-lambda}.
\begin{table}[!htbp]
\center
\caption{User-delivery array $\mathbf{Q}$.
\label{tab-user-delivery-lambda}}
\renewcommand\arraystretch{0.5}
\begin{tabular}{|c|ccccccc|}
Subpacket labels&\multicolumn{7}{|c|}{User set $\mathfrak{B}$} \\ \hline
$\mathcal{D},\mathcal{T}$&$U_{1235}$ &$U_{2346}$ &$U_{3457}$ &$U_{1456}$ &$U_{2567}$ &$U_{1367}$ &$U_{1247}$ \\ \hline
$\{1\},\{1,2\}$ & * & 123,1 & 134,2 & * & 125,1 & * & * \\
$\{2\},\{1,2\}$ & * & * & 234,1 & 124,1 & * & 123,1 & * \\
$\{3\},\{1,2\}$ & * & * & * & 134,1 & 235,1 & * & 123,1 \\
$\{4\},\{1,2\}$ & 124,1 & * & * & * & 245,2 & 134,2 & * \\
$\{5\},\{1,2\}$ & * & 235,1 & * & * & * & 135,1 & 125,1 \\
$\{6\},\{1,2\}$ & 126,1 & * & 346,1 & * & * & * & 126,2 \\
$\{7\},\{1,2\}$ & 127,1 & 237,2 & * & 147,1 & * & * & * \\\hline
$\{1\},\{1,3\}$ & * & 124,1 & 135,1 & * & 126,2 & * & * \\
$\{2\},\{1,3\}$ & * & * & 235,1 & 125,1 & * & 126,2 & * \\
$\{3\},\{1,3\}$ & * & * & * & 135,1 & 236,1 & * & 134,2 \\
$\{4\},\{1,3\}$ & 134,1 & * & * & * & 246,1 & 146,1 & * \\
$\{5\},\{1,3\}$ & * & 245,1 & * & * & * & 156,1 & 145,1 \\
$\{6\},\{1,3\}$ & 136,1 & * & 356,2 & * & * & * & 146,1 \\
$\{7\},\{1,3\}$ & 137,1 & 247,1 & * & 157,2 & * & * & * \\ \hline
$\{1\},\{1,4\}$ & * & 126,1 & 137,1 & * & 127,1 & * & * \\
$\{2\},\{1,4\}$ & * & * & 237,1 & 126,1 & * & 127,1 & * \\
$\{3\},\{1,4\}$ & * & * & * & 136,1 & 237,1 & * & 137,1 \\
$\{4\},\{1,4\}$ & 145,1 & * & * & * & 247,1 & 147,1 & * \\
$\{5\},\{1,4\}$ & * & 256,1 & * & * & * & 157,1 & 157,2 \\
$\{6\},\{1,4\}$ & 156,1 & * & 367,1 & * & * & * & 167,1 \\
$\{7\},\{1,4\}$ & 157,1 & 267,1 & * & 167,1 & * & * & * \\\hline
$\{1\},\{2,3\}$ & * & 134,1 & 145,1 & * & 156,1 & * & * \\
$\{2\},\{2,3\}$ & * & * & 245,1 & 245,2 & * & 236,1 & * \\
$\{3\},\{2,3\}$ & * & * & * & 345,1 & 356,2 & * & 234,1 \\
$\{4\},\{2,3\}$ & 234,1 & * & * & * & 456,1 & 346,1 & * \\
$\{5\},\{2,3\}$ & * & 345,1 & * & * & * & 356,2 & 245,2 \\
$\{6\},\{2,3\}$ & 236,1 & * & 456,1 & * & * & * & 246,1 \\
$\{7\},\{2,3\}$ & 237,1 & 347,1 & *	& 457,1 & *	&	*	&	*	\\	\hline
$\{1\},\{2,4\}$	&	*	&	136,1	&	147,1	&	*	&	157,2	&	*	&	*	\\	
$\{2\},\{2,4\}$	&	*	&	*	&	247,1	&	246,1	&	*	&	237,2	&	*	\\	
$\{3\},\{2,4\}$	&	*	&	*	&	*	&	346,1	&	357,1	&	*	&	237,2	\\	
$\{4\},\{2,4\}$	&	245,1	&	*	&	*	&	*	&	457,1	&	347,1	&	*	\\	
$\{5\},\{2,4\}$	&	*	&	356,1	&	*	&	*	&	*	&	357,1	&	257,1	\\	
$\{6\},\{2,4\}$	&	256,1	&	*	&	467,1	&	*	&	*	&	*	&	267,1	\\	
$\{7\},\{2,4\}$	&	257,1	&	367,1	&	*	&	467,2	&	*	&	*	&	*	\\	\hline
$\{1\},\{3,4\}$	&	*	&	146,1	&	157,1	&	*	&	167,1	&	*	&	*	\\	
$\{2\},\{3,4\}$	&	*	&	*	&	257,1	&	256,1	&	*	&	267,1	&	*	\\	
$\{3\},\{3,4\}$	&	*	&	*	&	*	&	356,1	&	367,1	&	*	&	347,1	\\	
$\{4\},\{3,4\}$	&	345,1	&	*	&	*	&	*	&	467,1	&	467,2	&	*	\\	
$\{5\},\{3,4\}$	&	*	&	456,1	&	*	&	*	&	*	&	567,1	&	457,1	\\	
$\{6\},\{3,4\}$	&	356,1	&	*	&	567,1	&	*	&	*	&	*	&	467,2	\\	
$\{7\},\{3,4\}$	&	357,1	&	467,1	&	*	&	567,1	&	*	&	*	&	*	\\	\hline
\end{tabular}
\end{table}
It is not difficult to check that $\mathbf{Q}$ is a $(7,42,24,42)$ PDA  which realizes a $(7,M',N)$ shared-link coded caching scheme with memory ratio $M'/N=4/7$ and transmission load $R=1$\footnote{With the same user number and memory ratio, by Lemma \ref{le-MN} we have the $(7,M',N)$ MN scheme with transmission load $R=3/5$ which is smaller than our scheme. }. Then we have a $(L,K,\Gamma,M,N)=(4,7,7,1,7)$ coded caching scheme under the $2$-$(7,4,7,2)$ design accessing topology with transmission load $R=1$.

Now let us show that the user-delivery array $\mathbf{Q}$ generated by \eqref{eq-array-user-PDA-lambda} satisfies C3 of Definition \ref{def-PDA}. Similar to the discussion of the case $\lambda=1$, we can easy to check that $\mathbf{Q}$ generated in \eqref{eq-array-user-PDA-lambda} satisfies Condition C$3$-a) of Definition \ref{def-PDA}. Now let us see Condition C$3$-b) in the following. Assume that there exist two
 distinct entries in $\mathbf{Q}$ which are filled by the same pair,
\begin{align}\label{eq-C3-proof-lambda}
\mathbf{Q}((\mathcal{T},\mathcal{D}),\mathcal{B})
=\mathbf{Q}((\mathcal{T}',\mathcal{D}'),\mathcal{B}')
=(\mathcal{D}\cup\mathcal{B}(\mathcal{T}),n_{\mathcal{D},\mathcal{B}(\mathcal{T})})
=(\mathcal{D}'\cup\mathcal{B}'(\mathcal{T}'),n_{\mathcal{D}',\mathcal{B}'(\mathcal{T}')}).
\end{align}If $\mathcal{D}\cup\mathcal{B}(\mathcal{T})
=\mathcal{D}'\cup\mathcal{B}'(\mathcal{T}')$ we have $\mathcal{D}= \mathcal{D}'$. Then from \eqref{eq-array-user-PDA-lambda} we have  $n_{\mathcal{D},\mathcal{B}(\mathcal{T})}\neq n_{\mathcal{D}',\mathcal{B}'(\mathcal{T}')}$, a contradiction to our hypothesis. So we have  $\mathcal{B}(\mathcal{T})\neq\mathcal{B}'(\mathcal{T}')$ which implies that $\mathcal{D}\neq \mathcal{D}'$.

Then for any $x\in \mathcal{D}\setminus \mathcal{D}'$ and $x'\in \mathcal{D}'\setminus\mathcal{D}$, we have
$$x\in \mathcal{B}'(\mathcal{T}') \ \ \ \text{and}\ \ \ x'\in \mathcal{B}(\mathcal{T})$$
which implies that
$\mathcal{B}'\bigcap \mathcal{D}\neq\emptyset$ and $\mathcal{B}\bigcap \mathcal{D}'\neq\emptyset$. From \eqref{eq-array-user-PDA-lambda} we have $\mathbf{Q}((\mathcal{T},\mathcal{D}),\mathcal{B}')$ $=\mathbf{Q}((\mathcal{T}',\mathcal{D}'),\mathcal{B})$ $=*$.

Similar to Subsection \ref{sub-perform}, let us consider the parameters $K$, $F$ and $S$. First we have $K=\lambda{\Gamma\choose t}/{L\choose t}$ from \eqref{eq-value-K} and $F={\Gamma\choose \mu\Gamma}{L\choose t}$ from \eqref{eq-array-node-caching}. So we only need to count the number of different pairs occurring in $\mathbf{Q}$ generated in \eqref{eq-array-user-PDA-lambda}, i.e., the value of $S$. Clearly there are at most $\lambda{\Gamma\choose t+\mu\Gamma}$ different pairs; for each $(t+\mu\Gamma)$-subset $\mathcal{S}$ belonging to some block $\mathcal{B}\in \mathfrak{B}$, $\mathcal{B}$ contains all the $t$-subsets which are the subset of $\mathcal{S}$, i.e., there are exactly ${\mu\Gamma+t\choose t}$ subsets $\mathcal{T}\in {[L]\choose t}$ satisfying $\mathcal{B}(\mathcal{T})\subseteq \mathcal{S}$, this implies that for each $t$-subset belonging to $\mathcal{S}$, the subset $\mathcal{S}$ occurs in $\mathbf{Q}$ at most $\lambda-1$ times, then there are at least $K{L\choose t+\mu\Gamma}$ pairs which do not occur in $\mathbf{Q}$. So $\mathbf{Q}$ has at most
\begin{align*}
S&\leq \lambda{\Gamma \choose t+\mu\Gamma}-K{L\choose t+\mu\Gamma}=\lambda{\Gamma \choose t+\mu\Gamma}-\frac{\lambda{\Gamma\choose t}}{{L\choose t}}{L\choose t+\mu\Gamma}\\
&=
\lambda\left({\Gamma \choose t+\mu\Gamma}-\frac{{\Gamma\choose t}{L\choose t+\mu\Gamma}}{{L\choose t}}
\right)
\end{align*}
different pairs in $\mathbf{Q}$. We have that $\mathbf{Q}$ generated in \eqref{eq-array-user-PDA-lambda} is a
$\left(K=\lambda{\Gamma\choose t}/{L\choose t}\right.$, $F={L\choose t}{\Gamma\choose \mu\Gamma}$, $Z={L\choose t}{\Gamma\choose \mu\Gamma}-{L\choose t}{\Gamma -L\choose \mu\Gamma}$, $\left.S\leq \lambda({\Gamma \choose t+\mu\Gamma}-{\Gamma\choose t}{L\choose t+\mu\Gamma})/{L\choose t}\right)$ PDA.

Recall $\mathbf{Q}$ listed in Table \ref{tab-user-delivery-lambda}. We can see that there exactly $K{L\choose t+\mu\Gamma}=7\times {4\choose 3}=28$ subset with size $3$
\begin{eqnarray}
\label{eq-once-subset}
\begin{split}
&\{1,2,3\},\{1,2,5\},\{1,3,5\},\{2,3,5\}\subseteq \{1,2,3,5\},&\\
&\{2,3,4\},\{2,3,6\},\{2,4,6\},\{3,4,6\}\subseteq \{2,3,4,6\},&\\
&\{3,4,5\},\{3,4,7\},\{3,5,7\},\{4,5,7\}\subseteq \{3,4,5,7\},&\\
&\{4,5,6\},\{1,4,5\},\{1,4,6\},\{1,5,6\}\subseteq \{1,4,5,6\},&\\
&\{5,6,7\},\{2,5,6\},\{2,5,7\},\{2,6,7\}\subseteq \{2,5,6,7\},&\\
&\{1,6,7\},\{3,6,7\},\{1,3,6\},\{1,3,7\}\subseteq \{1,3,6,7\},&\\
&\{1,2,7\},\{1,4,7\},\{2,4,7\},\{1,2,4\}\subseteq \{1,2,4,7\}.&
\end{split}
\end{eqnarray} each of which occurs in $\mathbf{Q}$ exactly once. Since the subset $\{1,2,6\}$ doest not occur in any blocks in \eqref{eq-once-subset} and $\lambda=2$, the subset $\{1,2,6\}$ occurs in $\mathbf{Q}$ exactly twice for each $t=2$-subset $\{1,2\}$, $\{1,6\}$ and $\{2,6\}$.

Finally similar to Subsection \ref{sub-perform}, we can also further reduce the transmission. First we can prove that Proposition \ref{pro-1} also holds for the $\mathbf{Q}$ generated by \eqref{eq-array-user-PDA-lambda}. Then given any block $\mathcal{B}\in \mathfrak{B}$, there are exactly $S'=\sum_{i\in [\mu\Gamma-1]}{L\choose t+i}{\Gamma-L\choose \mu\Gamma -i}$ $(t+\mu\Gamma)$-subsets $\mathcal{S}$ of $[\Gamma]$ satisfying $t+1\leq |\mathcal{S}\cap \mathcal{B}|<t+\mu\Gamma$, and for each $t$-subset of $\mathcal{S}$, there are exactly $\lambda$ different pairs containing $\mathcal{S}$, so by Proposition \ref{pro-1} the number of the coded packets which are unknown for user $U_{\mathcal{B}}$ is at most $S-\lambda S'$. Using a $[S,S-\lambda S']$ maximum sparable code, the server only needs to send $S-\lambda S'$ coded packets. Then the transmission load is
\begin{eqnarray}
\label{eq-load-lambda}
\begin{split}
R&= \frac{\lambda{\Gamma \choose t+\mu\Gamma}-\lambda\sum^{\mu\Gamma-1}_{i=1}{L\choose t+i}{\Gamma-L\choose \mu\Gamma -i}-K{L\choose t+\mu\Gamma}}{{L\choose t}{\Gamma\choose \mu\Gamma}}&\\[0.2cm]
&=\frac{\lambda{\Gamma \choose t+\mu\Gamma}}{{\Gamma\choose \mu\Gamma}{L\choose t}}-
\frac{\lambda\sum^{\mu\Gamma-1}_{i=1}{L\choose t+i}{\Gamma-L\choose \mu\Gamma -i}}{{L\choose t}{\Gamma\choose \mu\Gamma}}-
\frac{K{L\choose t+\mu\Gamma}}{{L\choose t}{\Gamma\choose \mu\Gamma}}.&
\end{split}
\end{eqnarray}

\begin{remark}\rm(The schemes via $t$-design with $\lambda=1$ and $\lambda>1$)
When there exists a $t$-$(\Gamma,L,K,1)$ design $(\mathcal{X},\mathfrak{B})$ and a $t$-$(\Gamma,L,\lambda K,\lambda)$ design $(\mathcal{X},\mathfrak{B}')$, we can obtain a $(L,K={\Gamma\choose t}/{L\choose t},\Gamma, M,N)$ coded caching scheme (Scheme A) under $(\mathcal{X},\mathfrak{B})$ with transmission load $R_1$ in \eqref{eq-load}, and a $(L,\lambda K,\Gamma, M,N)$ coded caching scheme (Scheme B) under $(\mathcal{X},\mathfrak{B}')$ with transmission load $R_2$ in \eqref{eq-load-lambda}. We can see that Scheme B has the $\lambda$ times larger transmission load than that of Scheme A. So by grouping method (replicating Scheme A $\lambda$ times) we can easily obtain a scheme with the user number $\lambda K$ which has the same performance as Scheme B.

We should point out that if there exists a $t$-$(\Gamma,L,K,1)$ design, then there must exit a $t$-$(\Gamma,L,\lambda K,\lambda)$ design for any $\lambda>1$. However it does not always hold conversely. So studying a coded caching scheme under a $t$-$(\Gamma,L,\lambda K,\lambda)$ design for any $\lambda>1$ is also meaningful.
\end{remark}

\section{Conclusion}
\label{sec:conclu}
In this paper, we considered a very general MACC setting where the topology defining the connectivity between the users and the caches is directly derived from a general combinatorial $t$-design and $t$-GDD structures. For this general setting, we have developed new schemes that manage to unify a broad spectrum of existing MACC schemes and the shared-link schemes. Most importantly, by allowing this unification, our schemes can dramatically densify the range of users for which high-performance MACC scheme can treat.

\begin{appendices}
\section{Proof of Lemma~\ref{le-cross-GDD}}
\label{sec:duality}
Let $(\mathcal{V}, \mathfrak{A})$ be a $t$-cross $(v,k,\Gamma_t,m,\lambda_t)$ resolvable design. Then $|\mathcal{V}|=v$, $|\mathfrak{A}|=\Gamma_t$, where each block has $k$ points. Let $q=v/k$. Without loss of generality let $\mathcal{V}=\{1,2,\ldots,v\}$ and $\mathfrak{A}$ consist of $m$ parallel classes $\mathfrak{A}_1$, $\mathfrak{A}_2$, $\ldots$, $\mathfrak{A}_m$ where
$$\mathfrak{A}_u=\{\mathcal{A}_{u,1},\mathcal{A}_{u,2},\ldots,\mathcal{A}_{u,q}\},\ \ \ \ \  u\in[m].$$
Then we can obtain the dual design $(\mathcal{X},\mathfrak{B})=(\mathfrak{A},\mathcal{V})$ where for each $j\in \mathcal{V}$, the block is defined as follows.
\begin{align}\label{eq-dual-block}
\mathcal{B}_j=\{\mathcal{A}_{u,v}\ |\ j\in \mathcal{A}_{u,v},u\in [m], v\in [q]\}
\end{align}Clearly each block $\mathcal{B}_j$ has exactly $m$ points, i.e., $|\mathcal{B}_j|=m$. Let
\begin{align}\label{eq-dual-group}
\mathfrak{G}=\{\mathcal{G}_1=\mathfrak{A}_1,\mathcal{G}_2=\mathfrak{A}_2,\ldots,
\mathcal{G}_{m}=\mathfrak{A}_{m}\}.
\end{align}
Clearly, each point occurs exactly $r$ times since each block of $\mathfrak{A}$ has $r$ points of $\mathcal{V}$, and $\mathfrak{G}$ is a partition of $\mathcal{X}$ into $m$ subsets each of which has size $q$.

Now let us show that the triple $(\mathcal{X},\mathfrak{G},\mathfrak{B})$ is a $t$-$(m,q,m,\lambda_t)$ GDD. By Definition \ref{def-cross-resolvable-desgin}, i.e., the intersection of any $t$ blocks drawn from any $t$ different parallel classes has the same size $\lambda_t$, we have that any $t$-subset of points in $\mathcal{X}$ from $t$ different groups belongs to exactly $\lambda_t$ blocks. So $(\mathcal{X},\mathfrak{B})$ is a $t$-$(m,q,m,\lambda_t)$ GDD.

\section{Proof of Lemma~\ref{le-cross-OA}}
\label{sec:proof of OA lemma}
Let us continue to use the notations in Lemma \ref{le-cross-GDD}. Recall the $t$-$(m,q,m,\lambda_t)$ GDD $(\mathcal{X},\mathfrak{G},\mathfrak{B})$ generated by \eqref{eq-dual-block} and \eqref{eq-dual-group} in proof Lemma \ref{le-cross-GDD}. Similarly to defining the block of $\mathfrak{B}$ \eqref{eq-dual-block} we can define a $v\times m$ array $\mathbf{A}=(\mathbf{A}(j,u))_{j\in \mathcal{X},u\in [m]}$ where
\begin{align}\label{eq-crros-row}
\mathbf{A}(j,u)=v  \ \ \ \text{if}\ \ j\in \mathcal{A}_{u,v}\ \ \text{for some integer}\ v\in [q].
\end{align} It is not difficult to check that for any $t$ columns, each vector of $[q]^t$ occurs in exactly $\lambda_t$ rows by the property of $t$-$(v,k,\Gamma,m,\lambda_t)$ resolvable design, i,e., the size of the intersection of the $t$ blocks $\mathcal{A}_{1,a_{j,1}}$,  $\mathcal{A}_{2,a_{j,2}}$, $\ldots$, $\mathcal{A}_{m,a_{j,m}}$ is $\lambda_t$, which implies that for any $t$ columns,  each vector of $[q]^t$ occurs in exactly $\lambda_t$ rows. By Definition \ref{def-OA} $v=\lambda_t q^t$, $\Gamma=qm$ and each integer occurs exactly $\lambda_t q^{t-1}$ in each column. This is the reason why the cross resolvable design always has $v=\lambda_t q^t$ and each block has exactly $\lambda_t q^{t-1}$ points in \cite{KMR,MKR,MR,DR}.

Conversely suppose $\mathbf{A}=(\mathbf{A}(j,u))_{j\in [\lambda_t q^t],u\in [m]}$ is a OA$_{\lambda_t}(m,q,t)$. Let
\begin{align}
\mathcal{V}=\{\mathbf{A}(j,\cdot)\ |\  j\in [\lambda_t q^t]\} \label{eq-point-set}
\end{align} and $\mathfrak{A}=\{\mathcal{A}_{u,v}\ |\ u\in [m], v\in [q]\}$ where
\begin{align}
\mathcal{A}_{u,v}=\{j\in [F]\ |\  \mathbf{A}(j,u)=v \}
\label{eq-block}
\end{align}
Clearly each $\mathfrak{A}_{u}=\{\mathcal{A}_{u,v}\ |\  v\in [q]\}$ where $u\in [m]$ is a parallel class of $\mathcal{V}$. Furthermore, one can check that $(\mathcal{V}, \mathfrak{A})$ is a $t$ cross $(v,k,\Gamma_t,m,\lambda_t)$ resolvable design by the property of OA$_{\lambda_t}(m,q,t)$, i.e., for any $t$ columns,  each vector of $[q]^t$ occurs in exactly $\lambda_t$ rows which is equivalent to that the intersection of any $L$ blocks from different $t$ parallel classes contains exactly $\lambda_t$ points.

\end{appendices}

\bibliographystyle{IEEEtran}
\bibliography{reference}

\end{document}